\documentclass[sigconf,nonacm]{acmart}

\makeatletter
\let\@acmBadgeL@image\@empty
\let\@acmBadgeR@image\@empty
\makeatother

\renewcommand\footnotetextcopyrightpermission[1]{}

\makeatletter
\newif\if@restonecol
\makeatother
 
\usepackage[linesnumbered,ruled,vlined]{algorithm2e}
\SetKwComment{Comment}{// }{}
\usepackage{amsmath}

\usepackage{enumitem}
\usepackage{multicol}
\usepackage{multirow}
\usepackage{float}
\usepackage{subfig}
\usepackage{balance}
\usepackage{bm}
\usepackage[normalem]{ulem}
\usepackage[table]{xcolor}
\usepackage{makecell}
\usepackage{tabularx}
\usepackage{booktabs}
\usepackage{array}
\usepackage{xcolor}
\usepackage{pifont}
\usepackage{bbding}
\usepackage{graphicx}
\usepackage{CJKutf8}
\usepackage{marginnote}
\usepackage{hhline}

\newcommand{\sym}[2][1.25]{\raisebox{0.15ex}{\scalebox{#1}{#2}}}

\newcommand{\xmark}{\sym{\textcolor{red!70!black}{\ding{55}}}} 

\newcommand{\bstar}{\sym{\textcolor{blue}{\ding{72}}}}

\definecolor{deepPink}{RGB}{204,0,102}

\newtheorem{theorem}{\textbf{Theorem}}

\newcommand{\sift}{\texttt{SIFT}\xspace}
\newcommand{\deep}{\texttt{Deep}\xspace}
\newcommand{\msong}{\texttt{MSong}\xspace}
\newcommand{\glove}{\texttt{GloVe}\xspace}
\newcommand{\gist}{\texttt{GIST}\xspace}
\newcommand{\tin}{\texttt{Tiny}\xspace}

\newcommand{\ivfflat}{\texttt{IVFFlat}\xspace}
\newcommand{\crackivf}{\texttt{CrackIVF}\xspace}
\newcommand{\tribase}{\texttt{Tribase}\xspace}
\newcommand{\spann}{\texttt{SPANN}\xspace}
\newcommand{\quake}{\texttt{Quake}\xspace}
\newcommand{\spfresh}{\texttt{SPFresh}\xspace}

\newcommand{\ourMethod}{\texttt{CLIP}\xspace}

\newcommand{\ourIVF}{\texttt{IVF-CLIP}\xspace}
\newcommand{\ourHIVF}{\texttt{HIVF-CLIP}\xspace}
\newcommand{\ourLSM}{\texttt{LSM-IVF}\xspace}

\AtBeginDocument{%
  }

\usepackage{tcolorbox}
\tcbuselibrary{breakable}
\usepackage{marginnote}
\usepackage{framed}
\newtcolorbox{MetaBox}{
  breakable,
  colback=green!5,      
  colframe=green!30,    
  boxrule=0.5pt,       
  arc=2pt,             
  left=5pt, right=5pt, top=5pt, bottom=5pt, 
  fontupper=\small 
}

\newtcolorbox{R1Box}{
  breakable,
  colback=blue!5,      
  colframe=blue!30,    
  boxrule=0.5pt,       
  arc=2pt,             
  left=5pt, right=5pt, top=5pt, bottom=5pt, 
  fontupper=\small 
}

\newtcolorbox{R3Box}{
  breakable,
  colback=orange!5,      
  colframe=orange!30,    
  boxrule=0.5pt,       
  arc=2pt,             
  left=5pt, right=5pt, top=5pt, bottom=5pt, 
  fontupper=\small 
}

\newtcolorbox{R4Box}{
  breakable,
  colback=red!5,      
  colframe=red!30,    
  boxrule=0.5pt,       
  arc=2pt,             
  left=5pt, right=5pt, top=5pt, bottom=5pt, 
  fontupper=\small 
}

\definecolor{shadecolor}{rgb}{0.92,0.92,0.92}

\colorlet{BLUE}{blue}
\colorlet{ORANGE}{orange!80!black}
\colorlet{RED}{red}
\begin{document}
\begin{CJK*}{UTF8}{gbsn}

\title{CLIP: Lightweight Cosine-Law-Based Inverted-List Pruning for IVF-Based Vector Search}

\author{Yitong Song}
\affiliation{
  \institution{Hong Kong Baptist University}
  \country{}
}
\email{yitong_song@hkbu.edu.hk}

\author{Shuhang Lu}
\affiliation{
  \institution{Shanghai Jiao Tong University}
  \country{}
}
\email{ts1989@sjtu.edu.cn}


\author{Pengcheng Zhang}
\affiliation{
  \institution{Tencent Inc.}
  \country{}
}
\email{petrizhang@tencent.com}

\author{Jianliang Xu}
\affiliation{
  \institution{Hong Kong Baptist University}
  \country{}
}
\email{xujl@comp.hkbu.edu.hk}

\begin{abstract}
\label{sec:abstract}
Vector search has become a core component of modern multimodal retrieval systems. Among existing methods, the inverted file (IVF)–based methods are widely adopted. However, they are fundamentally limited by coarse-grained execution: each query typically probes many clusters and exhaustively scans all vectors within them, resulting in high query latency. Prior works mitigate this using pruning strategies, but they often incur substantial extra pruning overhead, lack cluster-level pruning, and compromise update efficiency due to heavy maintenance of pruning metadata.

This paper proposes \ourMethod, a lightweight cosine-law–based pruning technique that supports both inter- and intra-cluster pruning, substantially reducing unnecessary cluster and vector accesses with negligible overhead.
First, \ourMethod exploits the monotonicity of cosine-law–based lower bounds, enabling eliminating an undesirable cluster in $O(1)$ time and filtering batches of irrelevant vectors in logarithmic time in the list size, with a tight analytical guarantee.
Second, building on this, we develop two IVF variants: \ourIVF, which integrates \ourMethod into IVFFlat, and \ourHIVF, which extends it with a hierarchical structure for adaptive sub-cluster probing. 
Third, for dynamic workloads, we present \ourLSM, an LSM-inspired design that supports fast updates by deferring index maintenance to background compaction, and enables efficient queries via \ourMethod-based optimizations that eliminate costly level-by-level searches.
Extensive experiments show that \ourMethod variants achieve up to 78\% pruning and 69\% higher efficiency over static IVF baselines, while \ourLSM improves throughput by up to 141\% over dynamic IVF baselines with comparable update efficiency.
\end{abstract}

\maketitle

\section{Introduction}
\label{sec:introduction}

Vector search aims to retrieve the top-$k$ most similar vectors from large-scale, high-dimensional datasets under similarity metrics such as Euclidean distance, typically formulated as Approximate Nearest Neighbor (ANN) search. It has become a cornerstone of modern multi-modal retrieval systems~\cite{ADBV, BlendHouse, PinnerSage}, supporting a broad range of applications such as image-to-image retrieval~\cite{ImageRetrieval}, music recognition~\cite{NowPlaying}, and code assistants~\cite{trae}. 

To support efficient ANN queries, Inverted File (IVF)–based methods have been extensively developed~\cite{spann, meta-faiss, CrackIVF, Quake, spfresh}. These methods operate by clustering the entire vector dataset into partitions that group similar vectors together. Each cluster is associated with an inverted list recording its assigned vectors. At query time, the top-$nprobe$ closest clusters are first identified by comparing their cluster centroids with the query vector. The vectors within these clusters are then scanned to compute exact distances and yield the top-$k$ most similar results. The parameter $nprobe$ controls the query efficiency and accuracy trade-off: larger values generally improve accuracy but reduce efficiency. Owing to their practical scalability, update friendliness, highly parallelizable structure, and satisfactory query performance, IVF-based methods have been widely adopted in real-world vector retrieval systems~\cite{Milvus, MicroNN, Quake, meta-faiss, spann, spfresh, GaussDB}.

\begin{table}[t]
  \centering
  \caption{Comparison of cluster access, data access, and pruning complexity at 99\% recall, with cluster count set to 4,096, where $l$ denotes the inverted list length and $subK$ is a constant. Legend: \xmark = Not supported, \bstar = Best performance.}
  \vspace{-0.1in}
  \label{tab:intro}
  \resizebox{\linewidth}{!}{
  \begin{tabular}{c|c|c|c|c|c|c|c}
    \toprule
    \multirow{2}{*}{Methods / Metrics} 
    & \multicolumn{2}{c|}{\makecell[c]{Accessed\\Cluster Count}} 
    & \multicolumn{2}{c|}{\makecell[c]{Accessed\\Data Ratio}} 
    & \multicolumn{3}{c}{\makecell[c]{Pruning Complexity\\(per Cluster)}} \\
    \cline{2-8}
     & GloVe & Deep 
     & GloVe & Deep 
     & Clus.-Prun. & Vec.-Prun. & Extra Stor. \\
    \hline
    \hline
    \makecell[c]{IVFFlat~\cite{meta-faiss}\\(traditional)} 
    & 500 & 400 
    & 35\% & 10\% 
    &  \multicolumn{3}{c}{\xmark} \\
    \hline
    \makecell[c]{Tribase~\cite{Tribase}\\(other pruning)}
    & 500 & 400 
    & 34\% & 10\% 
    & \xmark
    & $O(l)$
    & \makecell[c]{$O(l +$\\$l \cdot subK)$} \\
    \hline

    \makecell[c]{IVF-CLIP\\(ours)}
    & {\bf 441 \bstar} & {\bf 215 \bstar} 
    & {\bf 21\% \bstar} & {\bf 4\% \bstar} 
    & $\boldsymbol{O(1)}$ \bstar
    & $\boldsymbol{O(2\log l)}$ \bstar
    & $\boldsymbol{O(l)}$ \bstar \\
    \bottomrule
\end{tabular}}
\end{table}

\begin{figure}[!t]
    \vspace{-0.15in}
    \centering
    \includegraphics[width=\linewidth]{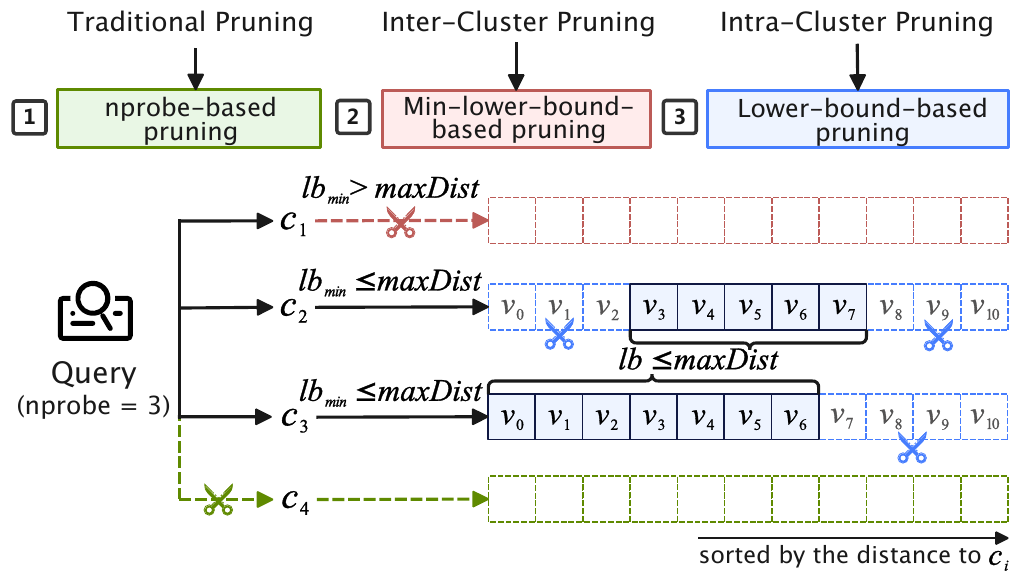}
    \vspace{-0.3in}
    \caption{Core pruning mechanism of CLIP.}
   \vspace{-0.25in}
    \label{fig:overview}
\end{figure}

Despite these advantages, IVF-based methods suffer from a fundamental bottleneck: they must access a large number of clusters and scan all vectors within them to perform expensive distance computations, leading to limited query efficiency.
As shown in Table~\ref{tab:intro}, achieving 99\% query recall requires the classic IVF-based method (\ivfflat) to probe 500 and 400 clusters, accessing 35\% and 10\% of the data vectors on the \glove and \deep datasets, respectively.
To alleviate this limitation, the recent work Tribase~\cite{Tribase} applies triangle-inequality–based pruning within each probed cluster, where all vectors in these clusters are evaluated individually and only the vectors that remain after pruning are then involved in distance computations.
However, Tribase still has several limitations:
(1) it offers no cluster-level pruning and still relies solely on $nprobe$ to control the number of accessed clusters; (2) it incurs substantial pruning overhead due to per-vector bound evaluation, resulting in $O(l)$ pruning complexity (where $l$ is the inverted list length) and $O(l+l \cdot subK)$ metadata storage;
(3) it shows low pruning ratios due to its rigorous triangle-inequality–based pruning conditions; and
(4) it undermines the update efficiency of IVF-based methods, as the widely-used in-place update strategy requires frequent cluster reorganization, with the stored metadata having to be re-computed accordingly, thereby increasing the online index maintenance cost.


Motivated by these, this paper proposes \ourMethod (\underline{\bf C}osine-\underline{\bf L}aw-based \underline{\bf I}nverted-list \underline{\bf P}runing), a lightweight method that enables both efficient and effective pruning in pure IVF structures. As illustrated in Figure~\ref{fig:overview}, beyond standard $nprobe$-based pruning, \ourMethod introduces two complementary strategies: inter- and intra-cluster pruning. Both rely on lower bounds derived from the Law of Cosines~\cite{cosineLaw}, which estimate the minimum possible squared distance between a query and a candidate vector. 
We detail them below.

\noindent
$\triangleright$
\uline{\textit{Inter-Cluster Pruning.}} Given $nprobe$ candidate clusters, we further identify pruning opportunities by estimating, for each cluster, the minimum lower bound \(lb_{\min}\) between the query vector \(q\) and any vector assigned to it. If \(lb_{\min}\) exceeds the current threshold \(maxDist\) (i.e., the squared distance of the current \(k\)-th nearest neighbor), the entire cluster can be safely pruned since all its vectors are guaranteed to be farther than \(maxDist\). Naïvely computing all individual lower bounds within a cluster to obtain \(lb_{\min}\) is expensive. Instead, we derive \(lb_{\min}\) in $O(1)$ time by exploiting the monotonicity of the cosine-law–based lower bound and closed-form expressions, enabling constant-time pruning. As shown in Table~\ref{tab:intro}, \ourMethod reduces the number of probed clusters from 500/400 to 441/215 on the two datasets, respectively, demonstrating effective cluster-level pruning.

\noindent
$\triangleright$
\uline{\textit{Intra-Cluster Pruning.}} For each candidate cluster that passes inter-cluster pruning, we retain only those in-cluster vectors whose $lb \leq maxDist$. Instead of computing the individual lower bound for each vector, we leverage the monotonicity of the bound again to efficiently locate all valid vectors.
Specifically, during index construction, each inverted list is sorted by centroid--vector distance, which is also stored with $O(l)$ extra space for query-time use. This induces a monotonic structure, implying that valid vectors form a contiguous segment within the sorted list.
At query time, we only need to locate the segment boundaries via two binary searches, retaining the in-range vectors and pruning the rest.
Overall, \ourMethod achieves vector-level pruning in $O(2\log l)$ time and substantially reduces data accesses (e.g., from 34\%/10\% to 21\%/4\% in Table~\ref{tab:intro}).

Beyond empirical results, we also provide theoretical analyses to explain the effectiveness of \ourMethod. In particular, we derive a quantile-based theoretical bound for the cosine-law–based lower bound, establishing its tightness in high-dimensional spaces and enabling principled control over the pruning tightness.
 
Building on \ourMethod, we develop two enhanced IVF variants, i.e., \ourIVF and \ourHIVF, based on the standard \ivfflat and the common hierarchical IVF method (HIVF)~\cite{spann, Quake}, respectively. Integrating \ourMethod into both variants requires only maintaining sorted inverted lists and storing centroid–vector distances, with negligible overhead from tightness parameters. \ourIVF is developed by directly integrating \ourMethod into \ivfflat, while \ourHIVF is further equipped with an adaptive query algorithm that selectively probes promising sub-clusters at each layer based on the efficiently estimated $lb_{\min}$ provided by \ourMethod. 
Our adaptive yet low-cost query algorithm improves upon existing approaches that either rely on simple heuristics to probe sub-clusters within a fixed distance range~\cite{spann} or employ costly models to estimate the likelihood that a sub-cluster contains query results~\cite{Quake}.

When handling dynamic workloads with both updates and queries, integrating \ourMethod improves query efficiency but increases update cost, since inverted list ordering and centroid–vector distances must be maintained. To address this, we propose \ourLSM, an LSM-inspired design that supports efficient online updates via a multi-level structure for fast in-memory writes, while deferring expensive index maintenance to background processing and keeping \ourMethod’s sorted lists through LSM-style merge sort operations. However, its multi-level structure requires level-by-level search and aggregation, which degrades query efficiency. 
We further eliminate these redundant searches by jointly identifying global top-$nprobe$ clusters and top-$k$ vectors across all \ourLSM levels in a single search procedure, and applying \ourMethod to remove unnecessary clusters and vectors.


To summarize, our contributions are as follows:
\vspace{-0.03in}
\begin{itemize}
\item We propose \ourMethod, a lightweight pruning method that enables (1) inter-cluster pruning to reduce the number of probed clusters, and (2) intra-cluster pruning to access only a small range of vectors within each inverted list. (Section~\ref{sec:clip})

\item We derive a quantile-based theoretical bound that explains the efficiency of \ourMethod and provides principled control over pruning tightness. The derived bounds are consistent with empirical observations. (Section~\ref{sec:clip-parameter})

\item Building on \ourMethod, we develop two efficient IVF variants, \ourIVF and \ourHIVF. In particular, \ourHIVF incorporates an adaptive query algorithm that selectively probes promising sub-clusters for each query, avoiding static heuristics and costly query-time estimation.  (Section~\ref{sec:integration})

\item To support dynamic workloads, we propose \ourLSM, an LSM-style design that utilizes \ourMethod-based optimizations to achieve both high update and query efficiency. (Section~\ref{sec:update})

\item Extensive experiments show that \ourIVF and \ourHIVF substantially outperform existing static IVF methods, achieving 69\% higher query efficiency and enabling pruning ratios of up to 78\% at the same recall.
\ourLSM improves throughput by up to 141\% over state-of-the-art dynamic IVF baselines with comparable update efficiency (Section~\ref{sec:exp}).

\end{itemize}

\section{Preliminaries}
This section formalizes the ANN query problem and reviews the IVF methods, providing necessary background for the paper.

\subsection{Problem Definition}
\label{sec:definition}
We begin with the exact $k$ nearest neighbor ($k$NN) query. Given a vector dataset $\mathcal{D}$ in $d$-dimensional space, a query vector $q$, and an integer $k$, the $k$NN query returns the $k$ vectors in $\mathcal{D}$ that are closest to $q$. Formally, it finds a subset $\mathcal{R}_{knn} \subseteq \mathcal{D}$ such that
$ |\mathcal{R}_{knn}| = k$ and $
\forall_{ v \in \mathcal{R}_{knn}}\forall_{o \in \mathcal{D} \setminus \mathcal{R}_{knn}}\Gamma(v,q) \leq \Gamma(o,q)$,
where $\Gamma(\cdot,\cdot)$ denotes a distance metric (e.g., Euclidean distance and inner product).
Due to the "curse of dimensionality", exact $k$NN search is prohibitively expensive in high-dimensional spaces. Recent research~\cite{HNSW, NSG, trim, DCO} focuses on the approximate nearest neighbor (ANN) search. By relaxing the accuracy requirement, ANN methods achieve substantial gains in efficiency. The accuracy is typically measured by the query recall, defined as
$Recall@k = |\mathcal{R}_{knn} \cap \mathcal{R}_{knn}'|/k$,
where $\mathcal{R}_{knn}$ and $\mathcal{R}_{knn}'$ represent the exact and approximate result sets, respectively.

\subsection{Review of Pure IVF Methods}
The inverted file (IVF) technique supports ANN search by partitioning the dataset into clusters using $k$-means~\cite{K-means} and maintaining an inverted list (also known as posting list) for each cluster to store its assigned vectors. During query processing, the query vector $q$ is first compared with all centroids to identify the $nprobe$ nearest clusters. Only vectors in these selected clusters are then scanned to compute distances to the query $q$, and the top-$k$ nearest vectors are returned. The parameter $nprobe$ controls the efficiency–accuracy trade-off: smaller values reduce computation but may degrade accuracy, while larger values improve accuracy at higher latency.

To further reduce search cost, the IVF index structure is often extended into \emph{hierarchical} or \emph{multi-layer} variants~\cite{spann,Quake}. These designs organize clusters hierarchically, either by recursively partitioning clusters into sub-clusters in a top-down manner or by recursively grouping existing clusters in a bottom-up manner. This hierarchical structure enables early pruning of irrelevant regions, thereby reducing both query-centroid distance comparisons and the number of scanned vectors during query processing.

\noindent
\textbf{Advantages and Limitations.} Due to their simple implementation, practical scalability, strong hardware parallelism (e.g., scanning independent inverted lists in parallel across CPU cores or GPUs), update friendliness (e.g., appending vectors to inverted lists), and satisfactory query performance, IVF methods are widely adopted in modern vector retrieval systems~\cite{meta-faiss, spann, Quake, MicroNN, Milvus, ADBV}. However, IVF structures provide only coarse-grained pruning: a large number of clusters must still be probed, and all vectors within the probed clusters need to be scanned during query processing, making this step a primary performance bottleneck. In addition, frequent updates can lead to imbalanced cluster sizes, potentially triggering costly online index maintenance operations such as re-partitioning~\cite{ada-ivf}.

\section{The \ourMethod Method}
\label{sec:clip}
This section presents \ourMethod to overcome the limitations of IVF indexes. We first provide an overview in Section~\ref{sec:clip-overview}, then describe the pruning strategy in Section~\ref{sec:clip-pruning}, and finally analyze its effectiveness and the parameters controlling pruning tightness in Section~\ref{sec:clip-parameter}.

\subsection{\ourMethod Overview}
\label{sec:clip-overview}
As illustrated in Figure~\ref{fig:overview}, in addition to the widely used $nprobe$-based pruning strategy that probes a fixed number of nearest clusters~\cite{meta-faiss, Tribase, CrackIVF, Quake}, \ourMethod further employs two pruning techniques, i.e., \emph{min-lower-bound-based pruning} for inter-cluster pruning and \emph{lower-bound-based pruning} for intra-cluster pruning. 
For each probed cluster, the min-lower-bound-based pruning strategy estimates the minimum lower bound (i.e., $lb_{min}$) on the squared distances between the query and all vectors in the cluster using the Law of Cosines~\cite{cosineLaw} and its monotonicity property. If $lb_{min}$ exceeds the threshold $maxDist$ (i.e., the squared distance of the $k$-th nearest neighbor found so far), the entire cluster can be safely skipped.
In contrast, lower-bound-based pruning operates at a finer granularity. Within a candidate cluster, it identifies a contiguous range of vectors whose lower bounds may fall below $maxDist$ and prunes all remaining vectors outside this range. The valid vector range can be efficiently determined via two binary searches to locate the start and end positions of qualifying vectors.

It is worth noting that, although \ourMethod relies on cosine-law–based lower bounds for pruning, both strategies are designed to operate \emph{without explicitly computing the lower bound} for each vector, thereby achieving high pruning efficiency. The next section presents the pruning strategies in detail.

\subsection{Pruning Strategies}
\label{sec:clip-pruning}
We first define the cosine-law-based lower bound and its monotonicity, and then show how it enables efficient pruning.

Let $c$, $q$, and $v$ denote the cluster centroid, query vector, and a data vector assigned to $c$, respectively. By the Law of Cosines~\cite{cosineLaw}, the squared distance between $q$ and $v$ can be written as:
\begin{equation*}
    \Gamma(q, v)^2 
    = \Gamma(q, c)^2 + \Gamma(c, v)^2 
      - 2 \cos\theta \, \Gamma(q, c)\Gamma(c, v),
\end{equation*}
where $\theta = \angle qcv$. Then, we provide the following theorem.

\begin{theorem}[Cosine-Law-Based Lower Bound]
\label{the:lb}
Let $\theta_{\min}$ denote the minimum possible angle among all admissible triplets $(q,c,v)$, and define $\lambda = \cos\theta_{\min}$. Then,
\begin{equation}
\label{equ:consineLB}
    lb_{qv}
    = \Gamma(q, c)^2 + \Gamma(c, v)^2 
      - 2\lambda \, \Gamma(q, c)\Gamma(c, v)
\end{equation}
is a valid lower bound of $\Gamma(q, v)^2$.
\end{theorem}

\begin{proof}
Since $\theta \ge \theta_{\min}$, we have $\cos\theta \le \lambda$. Substituting this into the cosine law yields the result.
\end{proof}

We defer the discussion of the bound’s validity in high-dimensional spaces and the choice of $\lambda$ to the next subsection, and now show how to exploit the monotonicity of this lower bound for efficient pruning without explicitly computing the bound for each vector. 

\begin{theorem}[Monotonicity of Cosine-Law-Based Lower Bound]
\label{the:monotonicity}
For fixed $q$, $c$, and $\lambda$, the cosine-law-based lower bound can be expressed as a function of $x = \Gamma(c,v)$:
\begin{equation}
\label{equ:lb}
    lb(x) = x^2 - 2\lambda\,\Gamma(q,c)\,x + \Gamma(q,c)^2
    = \big(x-\lambda \Gamma(q,c)\big)^2 + (1-\lambda^2)\Gamma(q,c)^2.
\end{equation}
Thus, $lb(x)$ is a convex quadratic function with a unique minimizer at $x^\star = \lambda\Gamma(q,c)$, attaining its minimum value $lb_{\min} = (1-\lambda^2)\,\Gamma(q,c)^2$, and is \emph{strictly decreasing} on $(-\infty, x^\star]$ and \emph{increasing} on $[x^\star, +\infty)$.
\end{theorem}

\begin{proof}
The result follows directly from the quadratic form.
\end{proof}

By sorting data vectors within each inverted list according to the centroid-vector distance, i.e., $\Gamma(c,v)$, and additionally storing these distances, we enable the following pruning strategies.


\begin{figure}[!t]
    \centering
    \includegraphics[width=\linewidth]{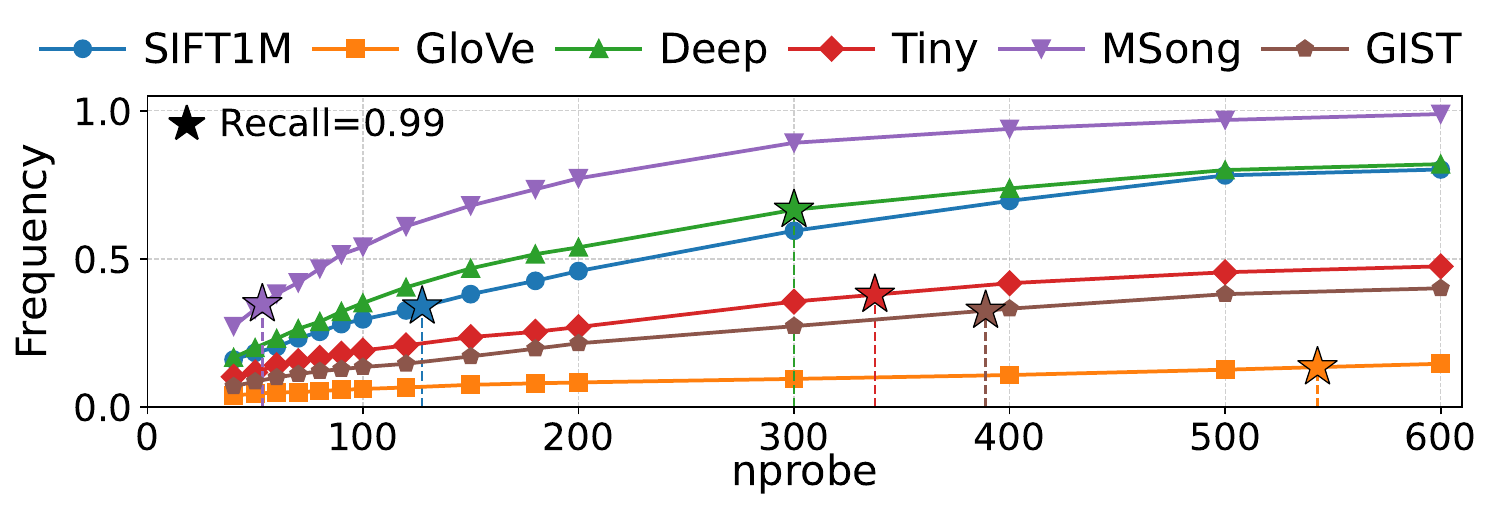}
    \vspace{-0.3in}
    \caption{Query frequency where the accessed cluster count inferred by inter-cluster pruning is smaller than $nprobe$.}
    \label{fig:nprobe}
\end{figure}

\noindent
{\bf Inter-Cluster Pruning (InterCP).} 
This pruning strategy eliminates unnecessary clusters among the $nprobe$ nearest candidate clusters. We first scan the data vectors in the nearest cluster to obtain initial top-$k$ results, and initialize $maxDist$ as the squared distance between the query and the current $k$-th nearest neighbor. 
Then, for each remaining cluster $c$, we compute a cluster-level minimum lower bound $lb_{\min}$, representing the smallest possible lower bound of the squared distance between the query and any vector in $c$, given by $lb_{\min} = (1-\lambda^2)\Gamma(q,c)^2$, as established in Theorem~\ref{the:monotonicity}. Here, $\lambda$ is precomputed, and $\Gamma(q,c)$ has already been obtained during the $nprobe$-based cluster selection phase. 
If $lb_{\min} > maxDist$, then for any vector $v \in c$, it holds that $\Gamma(q,v)^2 \ge lb_{qv} \ge lb_{\min} > maxDist$, implying that cluster $c$ cannot contribute to the top-$k$ result and can be pruned without accessing its list. This check incurs only $O(1)$ time per cluster and is applied to all $nprobe$ clusters.

Note that we first perform standard $nprobe$-based cluster selection and then apply InterCP within the selected $nprobe$ clusters. An alternative design is to rely solely on InterCP for cluster selection, i.e., accessing all clusters satisfying $lb_{\min} \leq \textit{maxDist}$ without explicitly specifying $nprobe$. 
However, this may lead to \emph{query-dependent} candidate clusters whose number is not fixed in advance, potentially increasing cluster access in certain cases. 
For example, as shown in Figure~\ref{fig:nprobe}, on the \deep datasets at 0.99 recall, using only InterCP, about $67\%$ of queries access fewer clusters than the $nprobe$-based method, while the remaining $33\%$ benefit from $nprobe$-based selection. Similar trends are observed across other datasets, indicating that neither strategy dominates uniformly across queries. 

We further observe that the clusters selected by InterCP differ from those selected by $nprobe$, reflecting their complementary nature: $nprobe$ relies solely on query–centroid distances, whereas InterCP leverages $lb_{\min}$, which accounts for all in-cluster vectors rather than only centroids.
These observations motivate a hybrid design that integrates $nprobe$-based selection with InterCP, combining strict control over the number of accessed clusters with adaptive query-dependent filtering. Importantly, this introduces almost no additional computational overhead, since all distances required by InterCP are already computed during the $nprobe$ phase. Empirically, applying InterCP after $nprobe$-based selection can further remove $6.7\%$–$46.2\%$ of the initially selected clusters (Section~\ref{sec:exp_static}), yielding an additional $10.7\%$–$29\%$ speedup (Section~\ref{sec:exp-clip}).

\noindent
{\bf Intra-Cluster Pruning (IntraCP).} 
This strategy prunes unnecessary data vectors inside each qualified cluster.
According to Equation~\ref{equ:consineLB}, the lower bound $lb_{qv}$ for each vector $v$ can be computed by re-using the query–centroid distance $\Gamma(q,c)$ and the pre-stored centroid–vector distance $\Gamma(c,v)$. 
A straightforward pruning approach would compute the lower bound for every vector in the cluster and compare it with $maxDist$ one by one, accessing only those with $lb_{qv} \leq maxDist$. 
However, this approach is inefficient since it requires $O(l)$ lower bound computations per cluster, where $l$ denotes the size of a probed cluster. Instead, we utilize the monotonicity of the cosine-law-based lower bound to identify qualified vectors without explicitly evaluating each lower bound. 

\begin{figure}[!t]
    \centering
    \includegraphics[width=\linewidth]{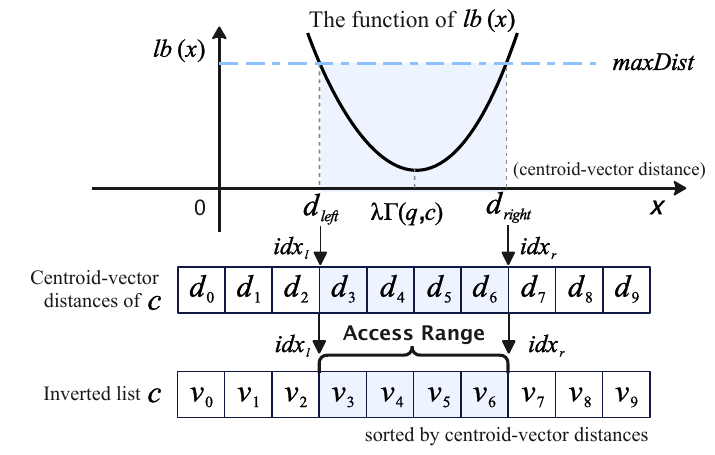}
   \vspace{-0.3in}
    \caption{Intra-cluster pruning of \ourMethod. The distances $d_{\text{left}}$ and $d_{\text{right}}$ are obtained as the two roots of $lb(x)=\textit{maxDist}$, where $lb(x)$ is defined in Equation~\ref{equ:lb}.
    Their corresponding positions $idx_l$ and $idx_r$ are located via two binary searches over the sorted centroid–vector distances. 
    Only vectors within the range $[idx_l, idx_r]$ are accessed.}
   \vspace{-0.1in}
    \label{fig:intra-cluster}
\end{figure}

Based on Theorem~\ref{the:monotonicity}, the lower bound first decreases and then increases with the centroid–vector distance $\Gamma(c,v)$, as shown in Figure~\ref{fig:intra-cluster}. Therefore, we can identify the two roots of $lb(x)=maxDist$, and any vectors whose centroid–vector distances fall between these two roots are guaranteed to have lower bounds smaller than $maxDist$. These vectors must be accessed for candidate evaluation. Since each inverted list is sorted in ascending order of centroid–vector distance, these vectors form a contiguous segment that can be efficiently located by identifying its boundary positions $\text{idx}_{\text{l}}$ and $\text{idx}_{\text{r}}$.

\begin{algorithm}[h]
\small
\caption{Intra-Cluster Pruning (IntraCP)}
\label{alg:icp}
\KwIn{Sorted distances $D_c=\{d_i|d_i = \Gamma(c,v_i), v_i \in \text{cluster c}\}$ for a candidate cluster $c$, $\Gamma(q,c)^2$, $\lambda$, $maxDist$}
\KwOut{Access range $[\text{idx}_{\text{l}}, \text{idx}_{\text{r}}]$}

Solve $lb(x)=maxDist$ to obtain $d_{\text{left}}$ and $d_{\text{right}}$ (Equation~\ref{equ:lb})\;
$\text{idx}_{\text{l}} \leftarrow \textsc{BinarySearch}(D_c, d_{\text{left}})$\; 
$\text{idx}_{\text{r}} \leftarrow \textsc{BinarySearch}(D_c, d_{\text{right}})$\;

\Return $[\text{idx}_{\text{l}}, \text{idx}_{\text{r}}]$
\end{algorithm}

Specifically, as illustrated in Algorithm~\ref{alg:icp}, the identification proceeds in three steps: (1) solving $lb(x)=maxDist$ to obtain two roots $d_{\text{left}}$ and $d_{\text{right}}$ ($d_{\text{left}} \le d_{\text{right}}$); (2) locating their corresponding positions $\text{idx}_{\text{l}}$ and $\text{idx}_{\text{r}}$ via binary search over the sorted centroid–vector distances; and (3) returning the index range $[\text{idx}_{\text{l}}, \text{idx}_{\text{r}}]$, within which all vectors satisfy $lb(x) \le maxDist$, while all remaining vectors are pruned without exact distance computation.

In practice, since nearby clusters are typically accessed, $\Gamma(q,c)$ is small and $d_{\text{left}}$ often becomes negative. In this case, $d_{\text{left}}$ is a mathematical artifact rather than a valid distance in the inverted list. Therefore, only a single binary search is needed to locate $d_{\text{right}}$, and all vectors with indexes from $0$ up to that position are accessed. In contrast, two binary searches are required only for far clusters where $d_{\text{left}} \ge 0$.

\subsection{Pruning Effectiveness and Parameters}
\label{sec:clip-parameter}
This section explains why \ourMethod is valid in high-dimensional spaces and how to set the key parameter $\lambda$ (defined as $\cos \theta_{\min}$). Here, $\theta_{\min}$ is the minimum possible value of the angle $\angle qcv$ over all feasible queries $q\in Q$, centroids $c$, and in-cluster vectors $v\in V_c$.

\begin{figure}[t]
  \centering
  \includegraphics[width=0.95\linewidth]{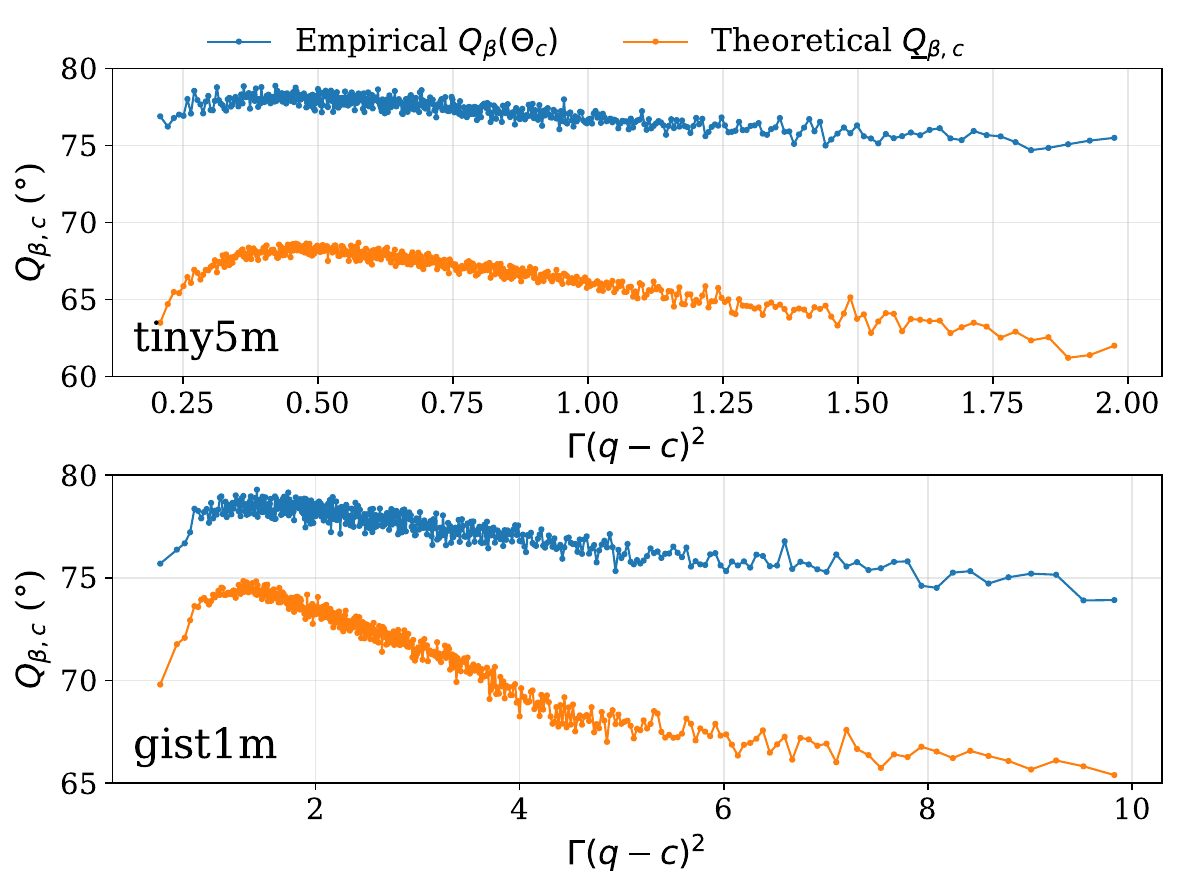}
  \vspace{-0.2in}
  \caption{Empirical $Q_\beta(\Theta_{c})$ values and their theoretical bound $\underline{Q}_{\beta,c}$ as a function of $\Phi=\Gamma(q-c)$, illustrating that (i) the bound consistently dominates the empirical values (Theorem~\ref{the:quantile:directional}), and (ii) it varies with $\Phi$, explicitly capturing the distance dependence (Theorem~\ref{the:quantile2}).}
  \label{fig:theta_min_vs_bound}
  \vspace{-0.2in}
\end{figure}

\noindent
\textbf{Pruning Effectiveness.}
\ourMethod employs the cosine-law–based lower bound for pruning, using a critical parameter $\theta_{\min} \leq \theta_{real}$ to ensure correctness. A larger $\theta_{\min}$ yields a tighter bound and thus stronger pruning. Since $\theta \in [0,\pi]$, $\theta_{\min}$ degenerates to $0$, reducing the bound to the traditional triangle-inequality–based form and making \ourMethod similar to prior methods such as iDistance~\cite{iDistance} that relies on distance ordering and triangle-inequality-based pruning. However, such bounds are proven to be weak in high-dimensional spaces due to the \emph{distance concentration} phenomenon, where pairwise distances become similar and the lower bounds approach 0~\cite{trim}.

In contrast, the cosine-law–based bound avoids this degeneracy with high probability due to a key high-dimensional property: random vectors are nearly orthogonal, i.e., pairwise angles concentrate sharply around $90^\circ$~\cite{Probability}. As a result, $\theta_{\min}$ typically remains close to $90^\circ$ rather than collapsing toward 0, leading to a tighter bound. 
Importantly, this concentration behavior of $\angle qcv$ (i.e., $\theta$) approximately persists even when $q$ and $v$ are not independent in the embedding space. This is because the angle is computed between residual vectors $\vec{cq}$ and $\vec{cv}$ rather than the original vectors. The centroid $c$ captures the dominant mean direction within the local cluster; subtracting $c$ therefore reduces the shared directional bias induced by the embedding distribution. Therefore, the residual vectors exhibit a more isotropic structure, making the high-dimensional angle concentration phenomenon still a reasonable approximation. 
Moreover, this concentration becomes stronger as the dimensionality increases, in the sense that the probability of $\theta_{\min}$ deviating significantly from $90^\circ$ decreases with dimension, typically at an exponential rate under standard concentration assumptions~\cite{Probability}.




We formalize this intuition in Theorem~\ref{the:quantile:directional}, which establishes a theoretical lower bound on the $\beta$-quantile of $\theta$, i.e., $Q_{\beta}(\Theta)$.  
We focus on $Q_{\beta}(\Theta)$, rather than $\theta_{min}$, for the following two reasons. First, computing the exact minimum angle would require exhaustively enumerating all feasible triplets $(q,c,v)$, which is computationally prohibitive. Second, the minimum value is statistically brittle: even a small number of outliers can drive it to an excessively low value, resulting in a loose bound. Instead, we consider a $\beta$-quantile of $\theta$ (e.g., the 0.001-quantile), which provides a robust and practically meaningful lower bound for pruning. Due to space constraints, the detailed proofs of all theorems in this section are provided in the accompanying technical report.\footnote{
\url{https://github.com/SongYitong826/CLIP/blob/main/Technical\%20Report.pdf}
}

\begin{theorem}[Lower Bound on $\beta$-Quantile of $\theta$]
\label{the:quantile:directional}
Given cluster centroid $c$, for any possible query $q\in Q$ and any in-cluster vector $v\in V_c$, define the angle set: 
\begin{equation*}
    \Theta_c := \{\angle qcv \in [0,\pi]\mid q\in Q,\ v\in V_c\}.
\end{equation*}

Let $Q_{\beta}(\Theta_c)$ denote the $\beta$-quantile of angles within $\Theta_c$. There exists a lower bound $\underline{Q}_{\beta,c}$ such that 
\begin{equation*}
Q_{\beta}(\Theta_c)  \ge  \underline{Q}_{\beta,c} = arccos
 \min \left\{
 1,
\left[B_{\beta}(c)
\sqrt{u^\top\Sigma_c u}.\right]
\right\},
\end{equation*}
where $u=\frac{q-c}{\|q-c\|}$ is the query direction, $\Sigma_c$ depends only on $c$, and $B_\beta(c)$ depends on $\beta$ and $c$.
\end{theorem}


Figure~\ref{fig:theta_min_vs_bound} empirically validates Theorem~\ref{the:quantile:directional}. On both datasets, the empirical quantile $Q_\beta(\Theta)$ is consistently bounded by the theoretical estimate, while the theoretical bound itself remains well separated from $0^\circ$ (with empirical values lying even further away). This confirms that \ourMethod does not degenerate to the triangle-inequality–based method and can provide stronger pruning effectiveness.

\noindent
\textbf{Parameter Determination.} Based on the above analysis, we set the key parameter $\lambda$ of \ourMethod to $\cos Q_{\beta}(\Theta)$. While this choice may incur a slight loss in accuracy, such a trade-off is well accepted in vector search, where a small loss in accuracy is often exchanged for a substantial improvement in query efficiency. 

To determine an effective value of $\lambda$, we make two key observations. First, although the theoretical bound closely tracks the empirical quantile and consistently provides a valid lower bound, it remains conservative and exhibits a non-negligible gap to empirical values, which limits pruning effectiveness. Second, as shown in both Theorem~\ref{the:quantile2} and Figure~\ref{fig:theta_min_vs_bound}, $Q_{\beta}(\Theta)$ depends on the query–centroid squared distance $\Gamma(q,c)^2$, implying that a single global $\lambda$ shared across all queries and clusters is inherently suboptimal.

\begin{theorem}[Distance Dependence of $Q_{\beta}(\Theta)$]
\label{the:quantile2}
Under the same setting as in Theorem~\ref{the:quantile:directional}, consider angles conditioned on the query--centroid distance $\Phi=\Gamma(q,c)$, denoted as $\Theta_c(\Phi):=\{\theta \in \Theta_c \mid \Gamma(q,c)=\Phi\}$.
Assume a locally anisotropic Gaussian model for the cluster distribution:
$\Sigma_c=\sigma_{\perp,c}^2 I + \Delta_c \hat{c}\hat{c}^\top$, $\hat{c} = \frac{c}{\|c\|}$,
where $\sigma_{\perp,c}^2$ and $\sigma_{\parallel,c}^2$ are the variances orthogonal and parallel to $\hat{c}$, respectively, and $\Delta_c=\sigma_{\parallel,c}^2-\sigma_{\perp,c}^2$. Then, the $\beta$-quantile of $\Theta_c(\Phi)$ admits a distance-dependent lower bound:
\[
Q_{\beta}(\Theta_c(\Phi)) \ge \underline{Q}_{\beta,c}(\Phi)
= \arccos\!\left(
\min\left\{
1,\;
B_{\beta}(c)\,
\sqrt{
\sigma_{\perp,c}^2
+
\Delta_c\, G(\Phi^2)
}
\right\}
\right),
\]
where $G(\Phi^2)$ is a non-decreasing function of $\Phi^2$ capturing the effect of query–centroid distance on the angular variance.
\end{theorem}


Based on these findings, we estimate $\lambda$ using empirical quantiles and make it adaptive to $\Gamma(q,c)^2$, instead of relying on the conservative theoretical bound. 
Specifically, we first sample a set of triplets $(q,c,v)$ and partition them into $p$ equal-width slices over $[\Gamma(q,c)_{\min}^2, \Gamma(q,c)_{\max}^2]$. When query vectors are unavailable, we approximate them using data vectors.Then, for each slice, we compute a slice-specific $\lambda = \cos Q_{\beta}(\Theta)$, where $Q_{\beta}(\Theta)$ is the $\beta$-quantile of $\theta$ within the sampled triplets falling into that slice.
During querying, the appropriate $\lambda$ is selected based on the slice index $\lfloor (\Gamma(q,c)_{\max}^2 - \Gamma(q,c)^2)/\textit{slice\_len} \rfloor$ for pruning, where $\textit{slice\_len} = (\Gamma(q,c)_{\max}^2 - \Gamma(q,c)_{\min}^2)/p$. This design replaces overly conservative theoretical bounds with empirical estimates, while capturing the dependence on $\Gamma(q,c)^2$ via distance-aware partitioning. Both $\beta$ and $p$ influence the query recall. Empirically, we observe that $\beta=0.001$ and $p=20$ consistently achieve high recall (Figures~\ref{fig:Parameter}).

\section{\ourMethod Integration}
\label{sec:integration}
In this section, we integrate \ourMethod into the IVFFlat structure (Section~\ref{sec:IVFCLIP}) and the hierarchical IVFFlat (HIVF) structure (Section~\ref{sec:HIVFCLIP}), yielding the variants \ourIVF and \ourHIVF, respectively.

\subsection{\ourIVF}
\label{sec:IVFCLIP}
When integrating \ourMethod into IVFFlat, both the index preparation and query processing pipelines require minor modifications.

\noindent
\textbf{Index Preparation.} \ourIVF follows the standard IVF construction pipeline, with the only modification being that inverted lists are sorted based on the centroid–vector distances. It also stores two auxiliary components: $(1)$ the sorted centroid–vector distances and $(2)$ the $\lambda$ values.
The centroid–vector distances are computed and stored directly during the clustering phase, incurring almost no extra computational cost. The values of $\lambda$ are obtained via the sampling method described in Section~\ref{sec:clip-parameter}, which is performed after the IVF index is built. As confirmed by our experiments (Figure~\ref{fig:index}), the overhead of computing these auxiliary statistics is negligible, as the clustering procedure itself dominates the index building cost. 

The additional storage overhead of \ourMethod is approximately $4n$ bytes for storing centroid–vector distances, where $n$ is the dataset size. Since the storage cost of $\lambda$ depends only on the number of slices $p$ (typically no more than 20), it is negligible and thus omitted in the subsequent space complexity analysis.

\noindent
\textbf{Query Processing.} As shown in Algorithm~\ref{alg:ivfclip}, \ourIVF begins by selecting the top-$nprobe$ clusters closest to the query vector $q$, identical to the standard IVF procedure (Line~3). For each probed cluster $c_i$, it locates the corresponding $\lambda$ and computes the minimum possible lower bound $lb_{\min}$ based on the already computed distance $\Gamma(q,c_i)^2$ (Lines~5--6). If $lb_{\min} > maxDist$, the entire cluster is discarded. Otherwise, \ourIVF proceeds with IntraCP (Lines~7--11): the two distance thresholds $d_{\text{left}}$ and $d_{\text{right}}$ are derived from Equation~\ref{equ:lb}, and their positions in the sorted inverted list are efficiently located via binary searches. Only vectors within this range are accessed for distance computations and updating the result queue $\mathcal{R}$. After all $nprobe$ clusters have been processed, $\mathcal{R}$ is returned.


\begin{algorithm}[t]
  \caption{\ourIVF($q$, $k$, $nprobe$)}
  \label{alg:ivfclip}
  \small
  \KwIn{query vector $q$, parameter $k$, probed cluster count $nprobe$}
  \KwOut{$k$ nearest neighbors of $q$}
  
  Result queue $\mathcal{R} = \emptyset$ (capacity $k$)\;
  $maxDist \leftarrow +\infty$ \tcp*{current $k$-th nearest distance}
  
  $C \leftarrow$ $nprobe$ nearest centroids \tcp*{$\Gamma(q,c_i)^2$ is computed here}
  \For{each centroid $c_i \in C$}{
    $\lambda \leftarrow \textsc{GetLambda}(\Gamma(q, c_i)^2)$ \tcp*{locate the slice}
    $lb_{\min} \leftarrow (1 - \lambda^2)\,\Gamma(q, c_i)^2$\;

    \If(\tcp*[f]{InterCP}){$lb_{\min} \le maxDist$}{
      Get sorted distances $D_i$ for cluster $c_i$\;
      $\text{idx}_{\text{l}}, \text{idx}_{\text{r}}$ = \text{IntraCP}($D_i$, $\Gamma(q, c_i)^2$, $\lambda$, $maxDist$) \tcp*{\footnotesize{Alg.~1}}
      \For{$t$ in $[\text{idx}_{\text{l}}\, \text{idx}_{\text{r}}]$}{
        \textsc{UpdateTopK}($\mathcal{R}$, $\Gamma(q, \text{cur\_list}[t])^2$, $maxDist$)\;
        }
    }
  }
  \Return{$\mathcal{R}$}\;
\end{algorithm}

\noindent
\textbf{Time Complexity.}
The additional time complexity introduced by \ourMethod is $O(nprobe + 2 \cdot nprobe \cdot \log l)$, where the $O(nprobe)$ term corresponds to computing $lb_{\min}$ for each candidate cluster, and the $O(2 \cdot nprobe \cdot \log l)$ term arises from the two binary searches performed on candidate inverted lists, with $l$ denoting the average list length. In contrast to existing pruning methods~\cite{trim, Tribase} that compute a lower bound for every candidate vector and incur a cost of $O(nprobe \cdot l)$, \ourMethod substantially reduces pruning overhead. 

\subsection{\ourHIVF}
\label{sec:HIVFCLIP}
HIVF organizes clusters hierarchically by recursively partitioning coarse clusters into finer sub-clusters. Query processing follows a top-down traversal, where a subset of clusters is explored at each layer and only vectors in the reached leaf clusters are accessed. A key challenge in HIVF is accurately identifying which sub-clusters should be probed. Existing methods either rely on simple heuristics that select sub-clusters within a predefined distance range~\cite{spann}, or adopt complex probabilistic models~\cite{Quake} to estimate the likelihood of a sub-cluster containing query results, often suffering from either suboptimal accuracy or high computational overhead.

In this section, we present \ourHIVF, which integrates HIVF with \ourMethod to adaptively and efficiently select sub-clusters at each layer using InterCP and IntraCP, and further reduces the number of scanned vector within leaf clusters via IntraCP. The index construction and query processing are detailed below.

\noindent
\textbf{Index Preparation.} \ourHIVF is constructed in a bottom-up manner, where finer clusters at a lower layer are recursively grouped to coarser clusters at higher layers. During this process, sub-clusters are ordered according to the distances between their centroids and their parent centroids, and these distances are stored for later use in computing lower bounds between queries and sub-cluster centroids.
After the HIVF index is built, the corresponding $\lambda$ values are computed and maintained at each layer. Similar to \ourIVF, the overhead of computing this extra information is negligible (Figure~\ref{fig:index}) as the clustering procedure dominates the index building cost. The additional storage cost is approximately $4 \times (n + n_c)$ bytes, where $n_c$ denotes the number of clusters across all HIVF layers.

\begin{algorithm}[t]
  \caption{\ourHIVF($q$, $k$, $nprobe$)}
  \label{alg:hivfclip}
  \small
  \KwIn{query vector $q$, parameter $k$, probed cluster count $nprobe$}
  \KwOut{$k$ nearest vectors of $q$}
  
  Result queue $\mathcal{R} = \emptyset$ (capacity $k$); 
  \tcp{\footnotesize $\mathcal{S}$ is sorted by the first key}
  Candidate cluster queue $\mathcal{S} = \{\text{tuple}(\text{root}.lb_{\min}, \text{root}, \Gamma(q,\text{root})^2)\}$\;
  $maxDist \leftarrow +\infty$ \tcp*{current top-$k$ vector dist.}
  $maxCD \leftarrow +\infty$ \tcp*{current top-$nprobe$ leaf cluster dist.} 
  $\text{cur\_{nprobe}} \leftarrow 0$ \tcp*{current accessed leaf cluster count} 
  \While{$\mathcal{S}$ is not empty and $\text{cur\_{nprobe}} <nprobe$}
  {
    $\text{cur\_tuple} \leftarrow$ $\mathcal{S}$.pop()\; 
    \If(\tcp*[f]{IntraCP}){$\text{cur\_tuple}$ denotes a leaf cluster}
    {
       Perform Lines~8--11 in Algorithm~\ref{alg:ivfclip}; $\text{cur\_{nprobe}++}$;
    }
    \Else
    {
       \If(\tcp*[f]{InterCP}){$\text{cur\_tuple}.lb_{\min} \leq maxCD$}
       {
          compute $d_{\text{left}}$ and $d_{\text{right}}$ based on $maxCD$ (Equation~\ref{equ:lb})\;
          $\text{idx}_{\text{l}} \leftarrow \textsc{BinarySearch}(\text{childList}, d_{\text{left}})$\;
          $\text{idx}_{\text{r}} \leftarrow \textsc{BinarySearch}(\text{childList}, d_{\text{right}})$\;
          \For(\tcp*[f]{IntraCP}){$t$ in $[\text{idx}_{\text{l}},\, \text{idx}_{\text{r}}]$}
          {
            child $\leftarrow \text{childList}[t]$\; 
            \If{child is a leaf cluster}
            {
            $\mathcal{S}$.add(tuple($\Gamma(q, \text{child})^2$, child, $\text{child}.lb_{\min}$))\;
            update $maxCD$\;
            }
            \Else
            {
            $\mathcal{S}$.add(tuple($\text{child}.lb_{\min}$, child, $\Gamma(q, \text{child})^2$))\;  
            }
          }
       }
    }
}
  \Return{$\mathcal{R}$}\;
\end{algorithm}

\noindent
\textbf{Query Processing.} 
As shown in Algorithm~\ref{alg:hivfclip}, we still retain $nprobe$ to control the number of accessed leaf clusters. Unlike \ourIVF, however, the top-$nprobe$ clusters are selected adaptively via hierarchical traversal instead of exhaustive scanning over all leaf clusters. 
During initialization (Lines 1--5), we initialize a result queue $\mathcal{R}$ of size $k$ and a cluster queue $\mathcal{S}$ seeded with the root cluster tuple, consisting of its $lb_{\min}$, cluster pointer, and squared distance to $q$. 
$\mathcal{S}$ is ordered by the \emph{first key}, which differs between leaf and non-leaf clusters: for leaf clusters, the key is the squared query-centroid distance, while for non-leaf clusters, it is the $lb_{\min}$ (as explained later). We maintain dynamic thresholds $maxDist$ and $maxCD$ for the current $k$-NN and top-$nprobe$ cluster distances, and track $cur\_{nprobe}$ as the number of accessed leaf clusters.
 
At each iteration (Lines 6--21), the algorithm pops the tuple in $\mathcal{S}$ with the smallest priority key. If the tuple corresponds to a \emph{leaf cluster}, \ourHIVF invokes the same IntraCP procedure as \ourIVF, accessing only a contiguous range of vectors whose lower bounds are below $maxDist$, and incrementing $cur\_{nprobe}$ by 1. If the tuple corresponds to a \emph{non-leaf cluster}, \ourHIVF applies both InterCP and IntraCP to determine the accessed sub-clusters, as shown in Figure~\ref{fig:HIVF} and detailed below: 

\noindent
$\triangleright$
\uline{\textit{Inter-Cluster Pruning (Line 11).}}
A non-leaf cluster can be safely pruned when the lower bound distances between $q$ and \emph{all} of its sub-cluster centroids exceed $maxCD$. This is detected by checking $lb_{min}$ of the current tuple, which represents the minimum lower bound distance between $q$ and any of its sub-clusters centroids. If $lb_{min} > maxCD$, none of its children can enter the top-$nprobe$ leaf cluster queue, and the cluster is discarded. 

\noindent
$\triangleright$
\uline{\textit{Intra-Cluster Pruning (Line 12--21).}} If the non-leaf cluster cannot be entirely pruned, \ourHIVF identifies a \emph{contiguous segment} of sub-clusters whose lower bound distances to $q$ are below $maxCD$, again using binary searches to determine the boundary indexes. For each qualified sub-cluster $child$, we compute both its exact squared distance and its minimum lower bound distance to $q$. If $child$ is a leaf cluster, it is inserted into $\mathcal{S}$ with its exact distance as the first key (Line~18); otherwise, it is inserted with $lb_{\min}$ as the first key (Line~21). This design ensures that a non-leaf cluster with $lb_{\min}$ smaller than the exact squared distance of any candidate leaf cluster is expanded first, since it may still contain a closer leaf cluster.

The search terminates once $\mathcal{S}$ becomes empty or $cur\_{nprobe} \geq nprobe$. At that point, the result queue $\mathcal{R}$ is returned.

\begin{figure}[!t]
    \centering 
    \includegraphics[width=\linewidth]{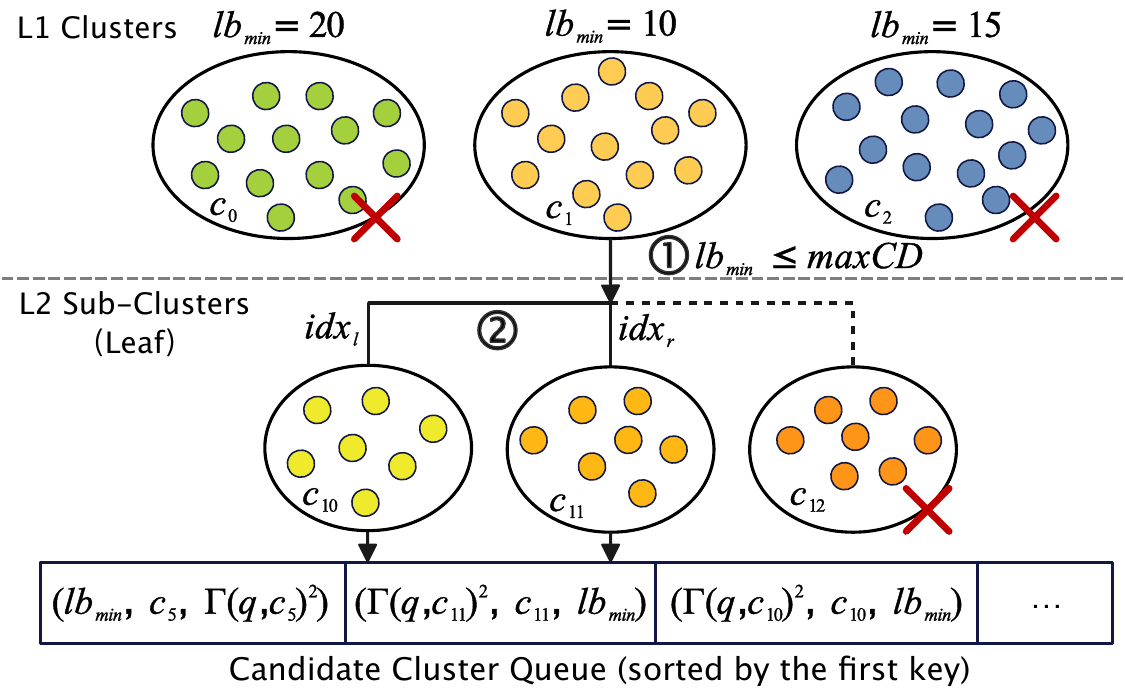}
    \vspace{-0.2in}
    \caption{The inter- and intra-cluster pruning in \ourHIVF.}
    \vspace{-0.15in}
    \label{fig:HIVF}
\end{figure}

\noindent
\textbf{Time Complexity.} 
Compared with \ourIVF, \ourHIVF adopts a hierarchical traversal strategy that reduces the cost of identifying the top-$nprobe$ clusters from $O(nlist)$ to $O(\log nlist)$, where $nlist$ denotes the number of leaf clusters. Empirically, \ourHIVF is more efficient than \ourIVF, as validated in Section~\ref{sec:exp_static}. The overall time complexity of \ourMethod under HIVF is $O(a_c + 2a_c \log s + 2 \cdot nprobe \cdot \log l)$, where $a_c$ is the average number of accessed non-leaf clusters with $a_c = \Theta(\log nprobe)$, and $s$ is the average branching factor of non-leaf nodes. The first two terms correspond to the InterCP and IntraCP costs at the non-leaf levels, while the last term accounts for IntraCP at the leaf level and dominates the overall complexity.

\section{Dynamic Workload Processing}
\label{sec:update}


Existing IVF methods support dynamic workloads mainly via in-place updates, i.e., directly modifying the IVF
index while serving queries~\cite{ada-ivf, spfresh, Quake}. 
Each update first identifies the target cluster, incurring an $O(nlist)$ cost without auxiliary routing indexes (e.g., HNSW~\cite{HNSW}), followed by lightweight list modifications (e.g., append-only insertions).
Therefore, update performance is dominated by the cluster identification phase.
For \ourMethod variants, the update overhead is further amplified since the ordering of inverted lists and distances must be additionally maintained.
Moreover, skewed updates can lead to severe cluster imbalance, often triggering costly online index maintenance such as rebalancing or rebuilding~\cite{Quake, spfresh}. 


To overcome these limitations, we propose \ourLSM, an LSM-style architecture that enables highly efficient updates by deferring costly maintenance to background and keeping the inverted lists sorted via an LSM-style merge sort process, while mitigating the search inefficiencies inherent to LSM-based designs through \ourMethod-based optimizations.
As shown in Figure~\ref{fig:update}, \ourLSM maintains a memory buffer at Level~0 and organizes IVF indexes across deeper LSM levels with increasing hierarchy: Level~1 uses single-layer \ourIVF, while Level~$i$ adopts an $i$-layer \ourHIVF. This design is motivated by the increasing data volume at deeper LSM levels, where the number of clusters grows accordingly, making multi-layer IVF structures more efficient at higher levels. We then detail how \ourLSM efficiently supports updates and queries.

\begin{figure}[!t]
    \centering
    \includegraphics[width=\linewidth]{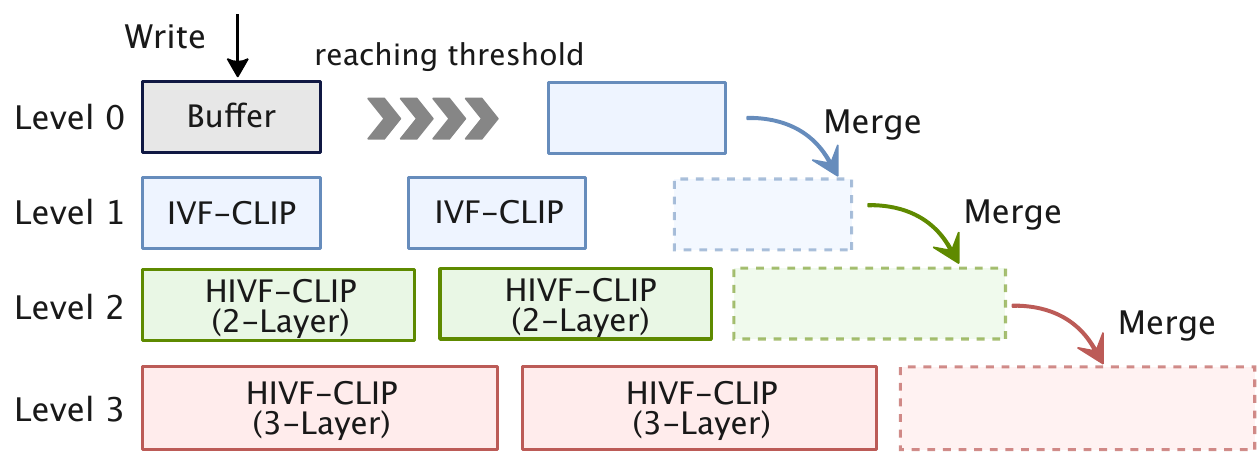}
    \vspace{-0.25in}
    \caption{The architecture of \ourLSM.}
   \vspace{-0.25in}
    \label{fig:update}
\end{figure}

\noindent
\textbf{Update Processing.} 
As in conventional LSM designs, newly inserted vectors are first buffered in $L_0$ and become immediately visible to queries. Once the buffer reaches its threshold, it is flushed and merged into Level~1. As levels fill up, merges are propagated to deeper levels via cascading compaction (e.g., $L_0 \rightarrow L_1 \rightarrow L_2 \rightarrow \cdots$). All merges are executed in the background, ensuring non-blocking query processing. 
During ongoing merges, queries are executed over a \emph{logical union} of all components (including the buffer and all existing levels), guaranteeing that no recently inserted vectors are missed, at the cost of temporarily increased query cost due to multi-level access.
Once a merge completes, the newly built component atomically replaces the corresponding old ones in this merge. 

The merge operation handles two cases: (i) merging unindexed vectors from $L_0$ into \ourIVF at $L_1$, and (ii) merging two IVF indexes across adjacent levels. In the first case, vectors from $L_0$ are assigned to the clusters of \ourIVF. For each cluster, newly assigned vectors are first collected into a list $list_{new}$, which is sorted by centroid–vector distances. This list is then merged with the existing list $list_{old}$ using a merge sort  process, following the standard compaction in LSM trees.
In the second case, vectors from the smaller index are reassigned into the larger index using the same merge sort process. For \ourHIVF, this merge is only restricted to leaf clusters.
To maintain clustering quality as data grows, we periodically trigger centroid re-training whenever a merge causes a level’s size to exceed predefined capacity thresholds (e.g., 25\%).



\noindent
\textbf{Query Processing.}
Unlike classic LSM-based search, which queries each level independently and then aggregates the results, \ourLSM flattens IVF indexes across all levels and performs a unified search. As shown in Algorithm~\ref{alg:lsm}, \ourLSM applies the \ourHIVF query procedure (i.e., Algorithm~\ref{alg:hivfclip}) with a modified initialization. 
Specifically, $\mathcal{R}$ is initialized with the top-$k$ results in the buffer (Line 1), and $\mathcal{S}$ is populated with candidate clusters from all levels (Lines 2--9): top-$nprobe$ clusters from the \ourIVF level using the exact distance as the priority key, and root clusters from \ourHIVF levels using their $lb_{\min}$ as the priority key.
The search then proceeds as in Algorithm~\ref{alg:hivfclip}, iteratively expanding the most promising cluster and applying \ourMethod for pruning until global $nprobe$ leaf clusters are accessed. In summary, \ourLSM improves query efficiency through:
(1) employing layer-appropriate IVF structures at different LSM levels, reducing cluster-identification cost as data volume grows; (2) jointly identifying the global top-$nprobe$ clusters across all levels within a single search process, avoiding level-by-level lookup and aggregation; and (3) using \ourMethod to prune both clusters and vectors, reducing data accesses and computations.

\begin{algorithm}[t]
  \caption{\ourLSM($q$, $k$, $nprobe$)}
  \label{alg:lsm}
  \small
  \KwIn{query vector $q$, parameter $k$, probed cluster count $nprobe$}
  \KwOut{$k$ nearest vectors of $q$}
  
  Result queue $\mathcal{R} = \textsc{Search}(\text{buffer}, q, k)$\;
  
  Candidate cluster queue $\mathcal{S} = \emptyset$; \tcp{sorted by the first key} 

    \For(\tcp*[f]{Adding clusters in all levels}){ $l$ in $(0,\, MaxL]$}
    {
       \If{$level[l].index ==$ \ourIVF}
       {
         $candidateC \leftarrow \textsc{TopKClusters}(q, nprobe)$ \;
         \For{each $c$ in $candidateC$}
         { 
             $\mathcal{S}$.add(tuple($\Gamma(q, c)^2$, $c$, $c.lb_{\min}$))\;
         }
       }
       \If{$level[l].index ==$ \ourHIVF}
       {
       $\mathcal{S}$.add(tuple($\text{root}.lb_{\min}$, $\text{root}$, $\Gamma(q, \text{root})^2$))\;
       }
    }

  Perform Lines 2–21 in Algorithm~\ref{alg:hivfclip}\;
  \Return{$\mathcal{R}$}\;
\end{algorithm}


\section{Experiments}
\label{sec:exp}
\subsection{Experimental Setup}
\label{sec:exp-setup}
\noindent \textbf{Datasets and Evaluated Workloads.}
As summarized in Table~\ref{tab:datasets}, we evaluate our methods on six widely used benchmark datasets for ANN search, including 
\glove~\cite{glove}, \sift~\cite{pq}, \deep~\cite{deep}, \tin~\cite{ANNdatasets}, \msong\cite{msong}, and \gist~\cite{pq}.
These datasets span diverse application domains, including text, image, and audio retrieval, with vector dimensionalities ranging from 128 to 960. For \sift, we use three scales: a 1M-vector subset for most experiments, a 10M-vector subset for scalability evaluation, and a 1B-vector subset for disk-based experiments. Table~\ref{tab:workloads} summarizes the evaluated workloads. In addition to a static, query-only workload, we consider three dynamic workloads with varying update-to-query ratios, modeling update-light, balanced, and update-intensive scenarios.

\begin{table}[t]
  \centering
  \vspace{-0.15in}
  \caption{Statistics of datasets}
  \vspace{-0.15in}
  \label{tab:datasets}
  \resizebox{\columnwidth}{!}{
  \begin{tabular}{c|c|c|c|c|c}
    \toprule
    Dataset & Dimension & \#Vectors & \#Queries & Data Size (GB) & Source \\
    \hline
    \hline

    \rowcolor{gray!15}
    SIFT
    & 128
    &
    \begin{tabular}{@{}c@{}}
      1,000,000 \\
      \hline
      10,000,000 \\
      \hline
      1,000,000,000
    \end{tabular}
    & 10,000
    &
    \begin{tabular}{@{}c@{}}
      0.48 \\
      \hline
      4.81 \\
      \hline
      119.21
    \end{tabular}
    & Images \\
    \hline

    GloVe 
    & 200 
    & 1,193,514 
    & 1,000 
    & 0.89
    & Texts \\
    \hline

    \rowcolor{gray!15}
    Deep
    & 256
    & 1,000,000
    & 1,000
    & 0.96
    & Images \\
    \hline

    Tiny
    & 384
    & 5,000,000
    & 1,000
    & 7.17
    & Images \\
    \hline

    \rowcolor{gray!15}
    MSong
    & 420
    & 994,185
    & 1,000
    & 1.56
    & Audios \\
    \hline

    GIST
    & 960
    & 1,000,000
    & 1,000
    & 3.58
    & Images \\
    \bottomrule
  \end{tabular}}
  \vspace{-0.1in}
\end{table}

\begin{table}[t]
\centering
\caption{Evaluated workloads}
\vspace{-0.4cm}
\label{tab:workloads}
\small
\setlength{\tabcolsep}{3pt}
\renewcommand{\arraystretch}{1.0}

\resizebox{0.92\linewidth}{!}{
\begin{tabular}{ccccc}
\toprule
\multicolumn{2}{c}{\textbf{Workload}} &
\textbf{Update : Query} &
\textbf{Scenario} \\
\midrule

\multicolumn{2}{c}{Static} &
0 : 10 &
Query-only \\

\midrule

\multirow{3}{*}{Dynamic}
 & Low  & 2 : 8 & Update-light \\
 & Medium  & 5 : 5 & Balanced \\
 & High & 8 : 2 & Update-intensive \\

\bottomrule
\end{tabular}
}
\end{table}

\begin{figure*}
    \centering
    \includegraphics[width=\linewidth]{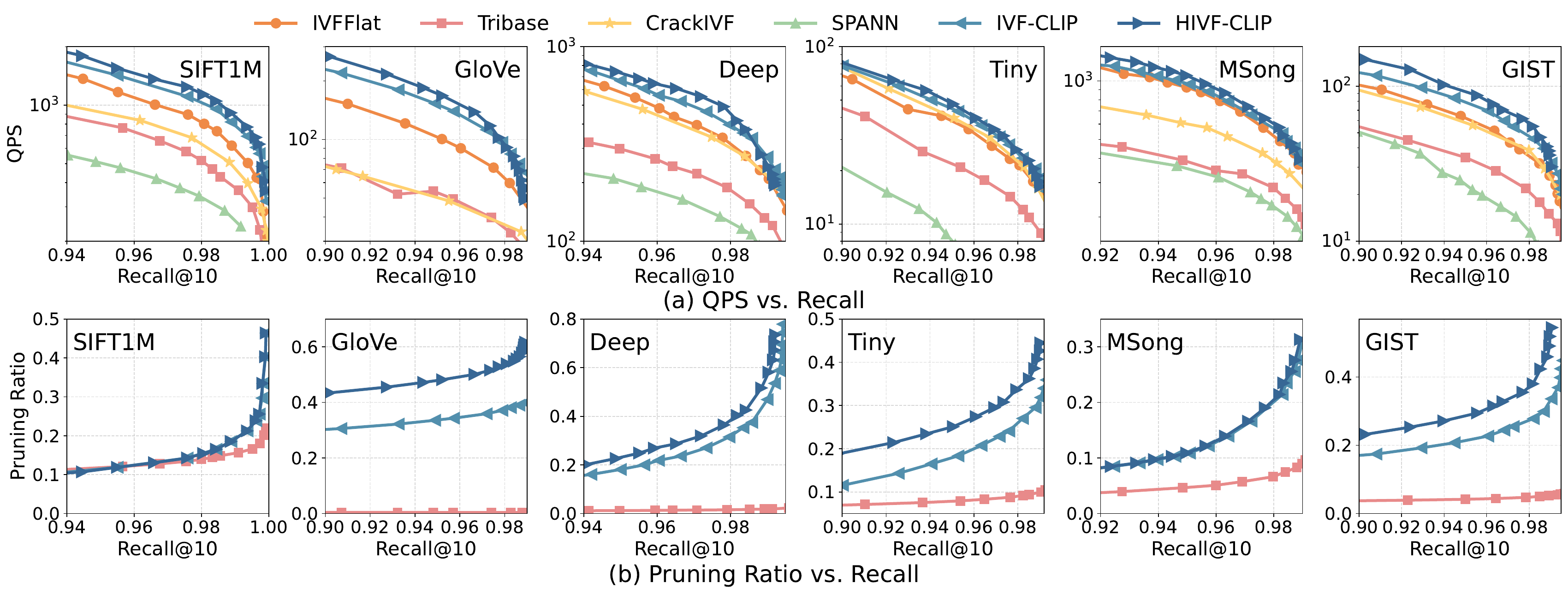}
    \vspace{-0.3in}
    \caption{Overall query performance under static workloads.}
    \vspace{-0.15in}
    \label{fig:staticperformance}
\end{figure*}

\noindent \textbf{Compared Methods.} We compare our methods against other pure IVF-based indexes under static and dynamic workloads.
Under static workloads that contain only queries, we compare \ourIVF and \ourHIVF with static single-layer methods, i.e., \ivfflat~\cite{meta-faiss}, \tribase~\cite{Tribase}, and \crackivf~\cite{CrackIVF}, as well as a multi-layer method, \spann~\cite{spann}.
Under dynamic workloads, we compare \ourLSM with SOTA dynamic methods, i.e., \spfresh~\cite{spfresh} and \quake~\cite{Quake}.
Descriptions of these baselines are detailed below.

\begin{itemize}
    \item \textbf{\ivfflat.} The standard IVF baseline implemented in Faiss~\cite{meta-faiss}, which partitions vectors into clusters and performs exhaustive vector scans within the probed lists without pruning.
    \item \textbf{\tribase.} An IVF-based pruning baseline that leverages the triangle inequality to derive distance- and angle-based bounds for each vector in the probed lists, eliminating those that do not satisfy the pruning criteria.
    \item \textbf{\crackivf.} An adaptive IVF baseline that incrementally reorganizes inverted lists during query processing to improve query performance.
    
    \item \textbf{\spann.} An HIVF baseline that constructs multi-layer balanced clusters and guides query processing from coarse to fine levels to reduce the search space.
    \item \textbf{\spfresh.} A dynamic HIVF baseline built on top of \spann, which maintains index quality under updates via several in-place strategies, i.e., online split and merge operations.
    \item \textbf{\quake.} 
    A dynamic HIVF method that adaptively reorganizes multi-layer partitions using a cost model. During query processing, it employs query-adaptive strategies (e.g., recall estimation) to dynamically select sub-clusters to access.
\end{itemize}

\noindent \textbf{Metrics.} 
Query efficiency and accuracy are measured using QPS (i.e., queries per second) and recall (i.e., $Recall@k$, defined in Section~\ref{sec:definition}), respectively.
Pruning effectiveness is quantified by the pruning ratio, defined as the reduction ratio of accessed data vectors before and after applying pruning.
To further illustrate pruning behavior, we additionally report the number of accessed clusters and the number of distance computations.
Throughput under dynamic workloads is evaluated by operation count per second.

\vspace{0.05in}
\noindent \textbf{Parameter Settings.} 
The number of clusters $nlist$ is set to 4,096 for all methods.
The key parameter $nprobe$ is dynamically adjusted to achieve different trade-offs between QPS and recall.
For \ourMethod variants, we set $\beta=0.001$ and use $20$ slices by default. For \ourHIVF, the number of cluster layers defaults to $2$.
For \ourLSM, the maximum data capacity of Levels~0,~1,~2, and~3 is set to 10K, 100K, 1M, and 10M vectors, respectively.
For all other baselines, we use the default method-specific parameters from their original implementations.

\vspace{0.05in}
\noindent \textbf{Implementations.}
Our experiments are conducted on a Linux server equipped with 512\,GB RAM, an Intel Xeon Platinum 8457C processor, and a 1.92\,TB Samsung enterprise SSD.
All methods are implemented in C++ with SIMD optimizations \textbf{\textit{opened}}. The baselines are implemented using their respective libraries~\cite{meta-faiss, sptaglib, Tribaselib, crackivflib, quakelib}. 
Following prior work~\cite{NSG, NHQ}, all available threads are used for index construction, while query processing is performed in a single-threaded, single-query setting.

\subsection{Performance under Static Workloads}
\label{sec:exp_static}

\noindent \textbf{Overall Query Performance.}
Figure~\ref{fig:staticperformance}(a) reports the QPS--Recall trade-offs for all static methods.
Across all datasets, \ourIVF and \ourHIVF consistently achieve the highest QPS at the same recall levels. On average, \ourIVF outperforms the strongest baseline in QPS by 23.4\%, 46.7\%, 20.5\%, 11.5\%, 5.8\%, and 29.4\% on \sift, \glove, \deep, \tin, \msong, and \gist, respectively. Building on this, \ourHIVF further improves over \ourIVF by an additional 11\%, 15\%, 14\%, 7\%, 10\%, and 20\% on average across the six datasets, benefiting from its hierarchical design that reduces candidate cluster identification cost.
Notably, at a high recall of 0.99, \ourIVF reaches up to a 51\% improvement over the strongest baseline, while \ourHIVF pushes the gain to as high as 69\%. 
These gains stem from reduced data access and computation costs enabled by \ourMethod. 

\noindent \textbf{Pruning Performance.} 
Figure~\ref{fig:staticperformance}(b) evaluates pruning effectiveness under different recall levels. This experiment includes only pruning techniques applicable to pure IVF indexes, i.e., \tribase and our methods. Over the recall range shown in the figure, \ourHIVF achieves an average pruning ratio 26.6$\times$ that of \tribase, with per-dataset multipliers of 7.4$\times$, 22.0$\times$, 123.1$\times$, 2.7$\times$, 1.06$\times$ and 3.5$\times$ respectively, reaching up to 78\% pruning over a recall of 0.99 on \deep dataset. 
We note that pruning effectiveness varies across datasets due to their intrinsic characteristics, particularly cluster distributions. Well-separated clusters yield tighter bounds and enable more aggressive pruning.
On the \sift dataset, although \ourIVF attains a $0.97\times$ pruning ratio of \tribase when recall is below 0.96, it still provides a $1.87\times$ QPS speedup over \tribase. 
This is because \ourMethod incurs substantially lower pruning overhead, as it avoids per-vector bound evaluation, whereas \tribase requires computing multiple costly bounds for each vector. As a result, \ourMethod achieves higher end-to-end query efficiency even when pruning fewer vectors.

It is worth noting that \tribase relies on the strict triangle inequality for pruning and guarantees no recall loss. Its rigorous pruning criterion makes it more suitable for extremely high recall regimes (e.g., $\geq 0.995$), where more candidates must be examined, creating greater opportunities for pruning. In contrast, \ourMethod adopts a lightweight, approximate pruning strategy that targets high efficiency with slight recall loss, and is particularly effective in moderate-to-high recall regimes (e.g., $\leq 0.995$).

\vspace{-0.1in}
\begin{figure}[h]
    \centering
    \includegraphics[width=0.95\linewidth]{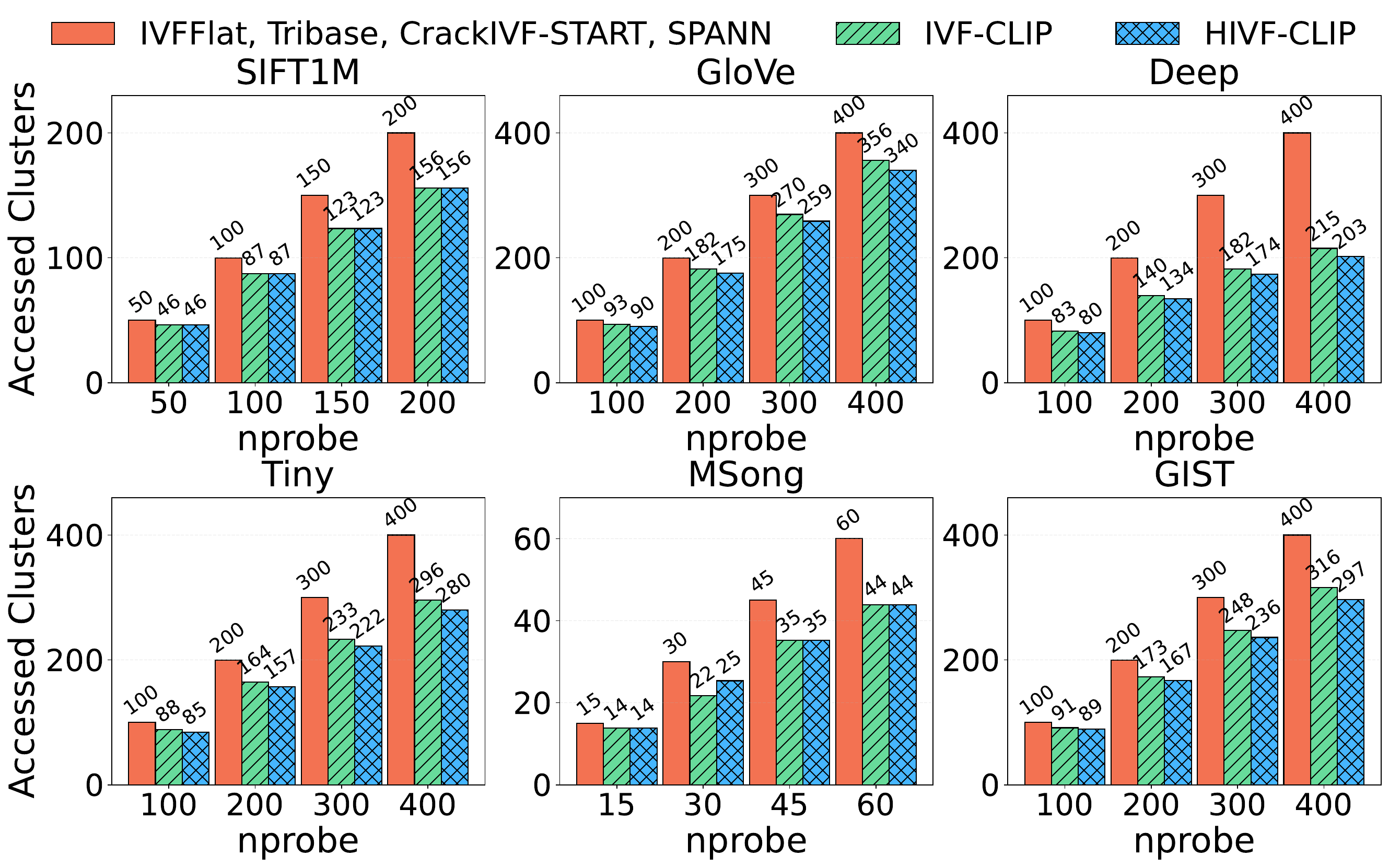}
    \vspace{-0.15in}
    \caption{Comparison of accessed cluster count.}
    \vspace{-0.1in}
    \label{fig:cluster}
\end{figure}

\noindent \textbf{Number of Accessed Clusters.} Figure~\ref{fig:cluster} compares the number of accessed clusters under different $nprobe$ settings. For \ivfflat, \tribase, \crackivf, and \spann, the number of accessed clusters is fully determined by $nprobe$. In contrast, \ourIVF and \ourHIVF leverage \ourMethod to prune the $nprobe$ candidate clusters, reducing the number of accessed clusters by 6.7\%–49.2\% across all datasets and $nprobe$ settings. Moreover, the pruning benefit becomes more pronounced as $nprobe$ increases: with larger $nprobe$ (typically required for higher recall), \ourMethod prunes a larger fraction of candidate clusters. We note that although \ourHIVF accesses a similar number of clusters as \ourIVF due to the shared pruning mechanism, it typically achieves better QPS–recall trade-offs (Figure~\ref{fig:staticperformance}) by further reducing the \emph{cost of identifying the initial top-$nprobe$ clusters} through hierarchical traversal (Section~\ref{sec:HIVFCLIP}). 


\noindent \textbf{Number of Distance Computations.} Figure~\ref{fig:distance} reports the number of distance computations as a function of $nprobe$. As $nprobe$ increases, all methods incur more distance computations. Nevertheless, \ourIVF and \ourHIVF consistently require the fewest distance computations, where the advantage over baselines is typically amplified at larger $nprobe$. 
Specifically, at the largest tested $nprobe$ shown in Figure~\ref{fig:distance}, \ourIVF requires only 0.75$\times$ the distance computations of the best baseline on average, ranging from 0.48$\times$ to 0.87$\times$ across all datasets, while \ourHIVF further reduces this ratio to 0.67$\times$, ranging from 0.43$\times$ to 0.87$\times$.
These results highlight the effectiveness of our InterCP and IntraCP strategies in lowering the number of distance computations.

\begin{figure}
    \centering
    \includegraphics[width=\linewidth]{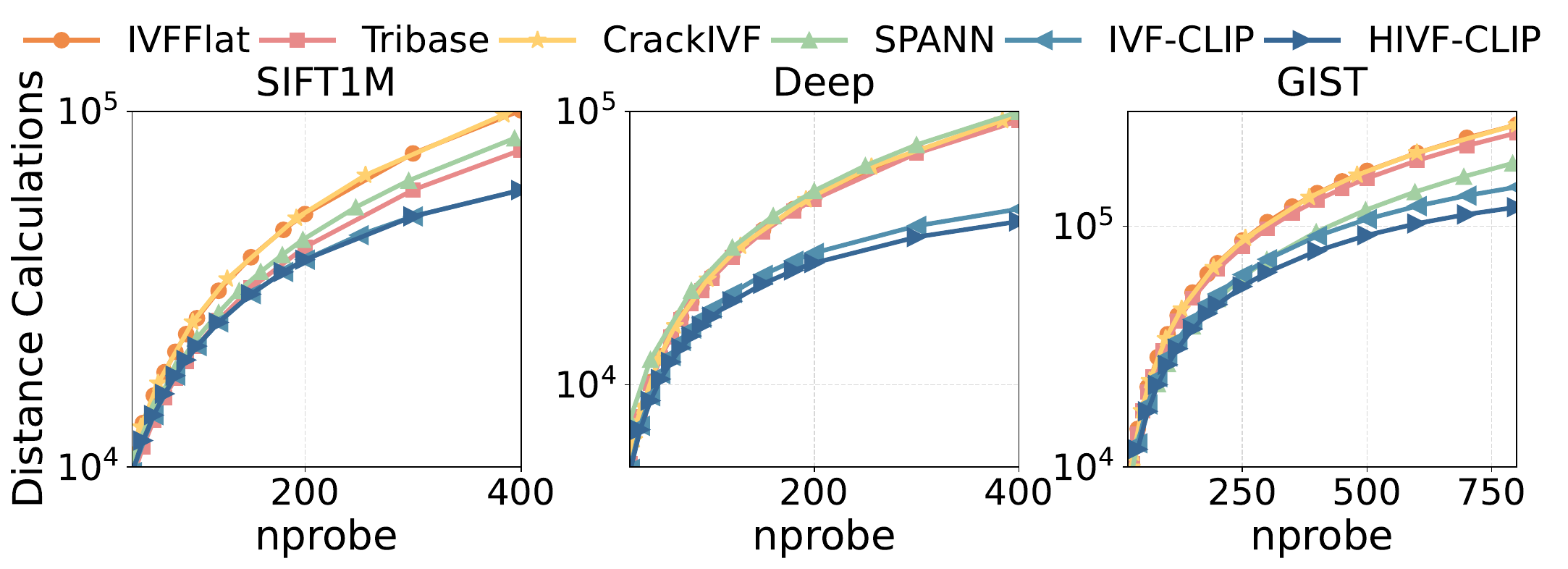}
    \vspace{-0.3in}
    \caption{Comparison of distance computation count.}
    \vspace{-0.1in}
    \label{fig:distance}
\end{figure}

\noindent
\textbf{Index Performance.}
Figure~\ref{fig:index} reports the index construction time, storage footprint, and peak memory usage during querying.
\ourIVF achieves construction time and storage cost comparable to \ivfflat, as it introduces only low computational cost (i.e., list sorting and sampling) and requires small extra storage for centroid–vector distances. Compared with \ourIVF, \ourHIVF requires 8.6\% more construction time and 0.4\% more storage cost, while remaining 73\% and 66\% less than other hierarchical method (i.e., \spann) in these two aspects, respectively. \spann incurs higher construction complexity and larger index size due to its hierarchical balanced clustering and the need to assign boundary vectors to multiple clusters for improving recall. \tribase typically incurs higher construction and storage overhead than \ourIVF, as it requires computing and storing multiple auxiliary neighbors per vector.

\vspace{-0.1in}
\begin{figure}[h]
    \centering
    \includegraphics[width=0.95\linewidth]{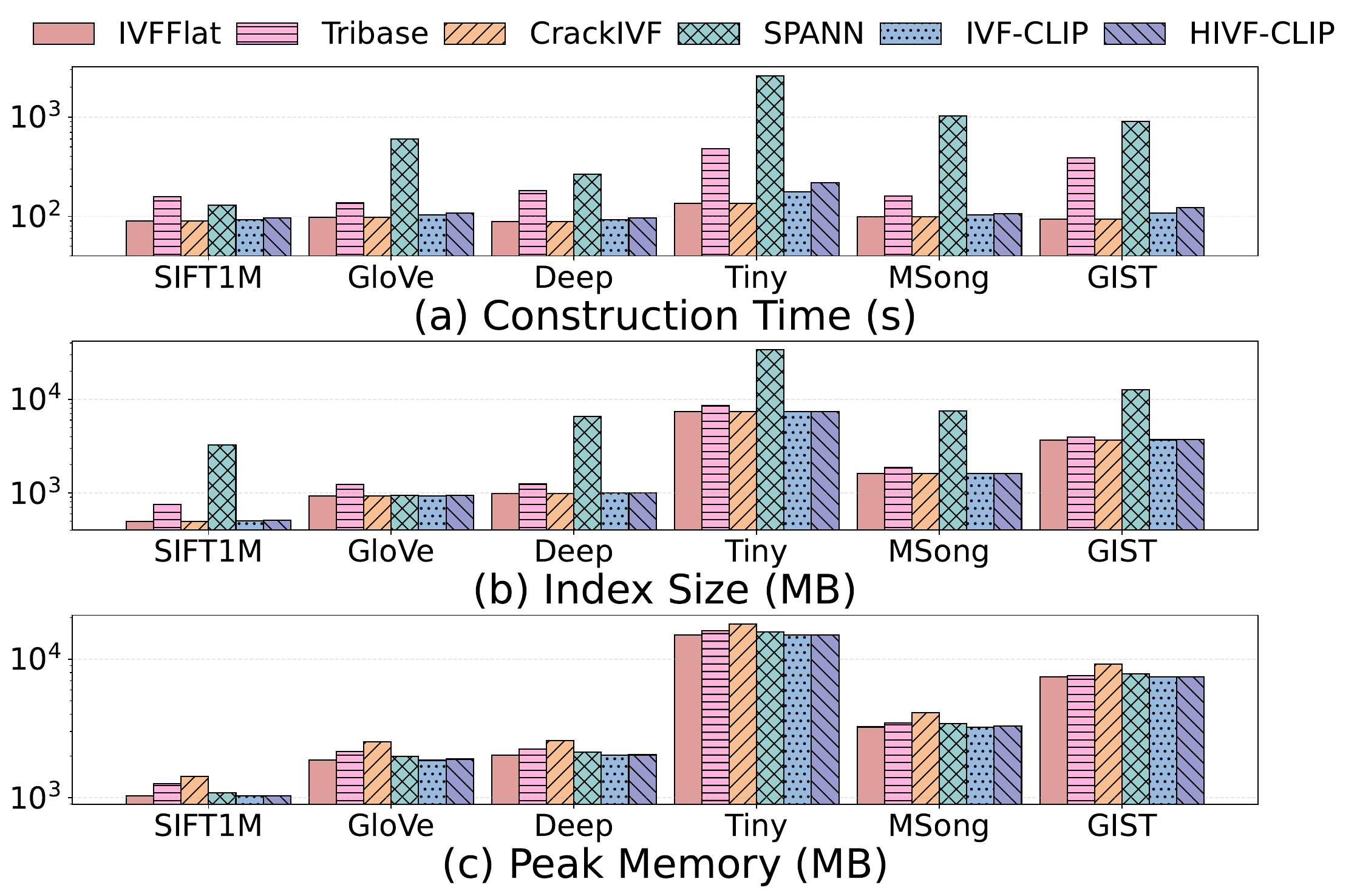}
    \vspace{-0.15in}
    \caption{Construction time, index size and peak memory.}
    \vspace{-0.2cm}
    \label{fig:index}
\end{figure}

In terms of peak memory during querying, most methods exhibit comparable memory footprints, as the accessed data volume dominates the memory overhead rather than the index structure itself. \ourIVF and \ourHIVF incur negligible additional memory overhead beyond standard IVF traversal, while \crackivf shows higher peak memory due to its adaptive cracking mechanism, which maintains additional intermediate state during query processing.

\begin{figure}[h]
    \centering
    \includegraphics[width=\linewidth]{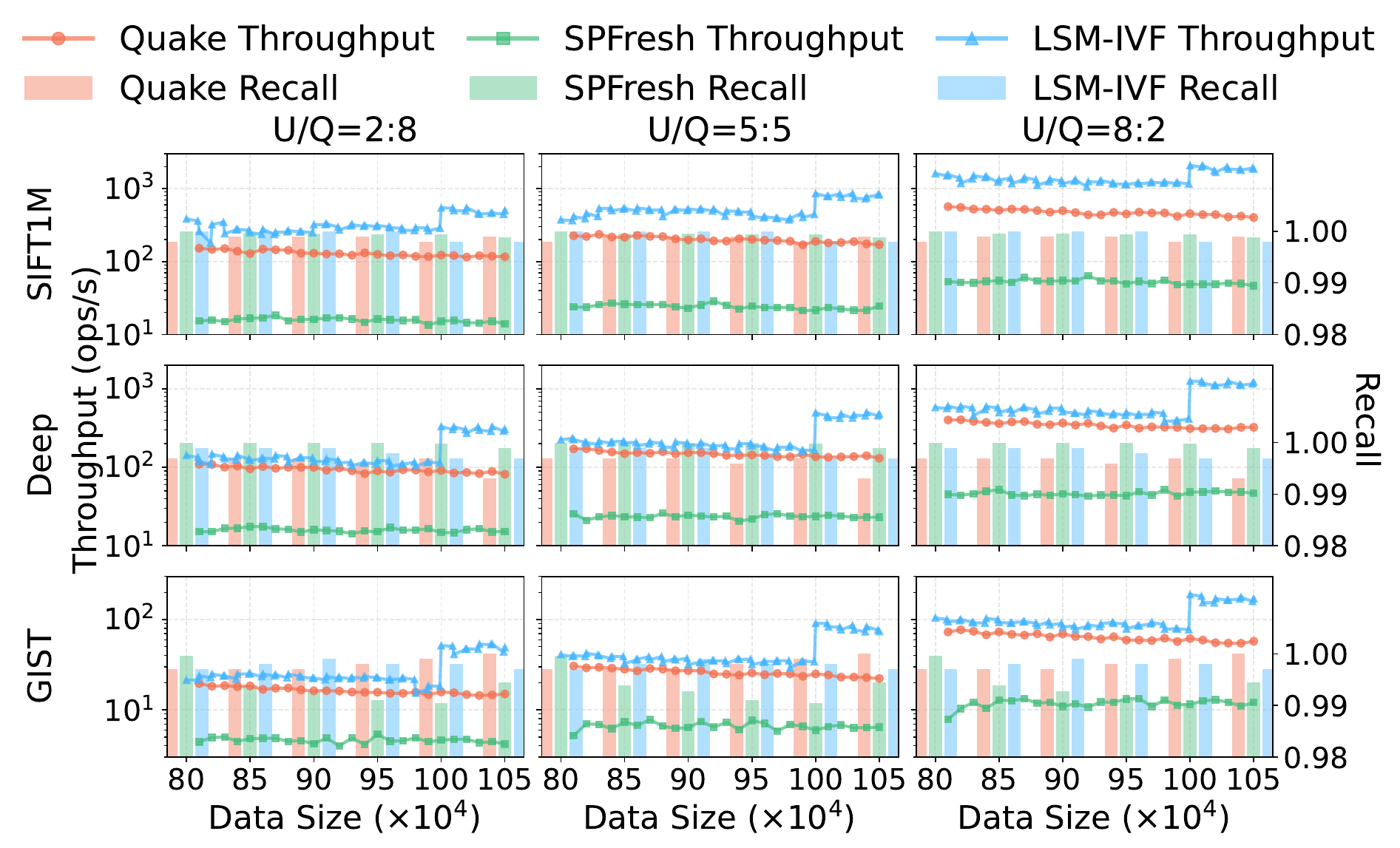}
    \vspace{-0.3in}
    \caption{Performance under varying dynamic workloads.}
    \label{fig:querywithupdates}
\end{figure}

\subsection{Performance under Dynamic Workloads}
\label{sec:exp-update}
\noindent
\textbf{Throughput and Recall.}
Figure~\ref{fig:querywithupdates} reports throughput and recall at $nprobe=400$ under different update-to-query ratios (U/Q). To evaluate \ourLSM after high-level merges, we extend each dataset beyond $10^6$ vectors through random sampling with an added bias. Across all workloads, \ourLSM consistently achieves the highest throughput. Before the dataset size reaches $10^6$, \ourLSM outperforms the strongest baseline (\quake) by 141\%, 38\%, and 39\% on \sift, \deep, and \gist, respectively. Averaged across datasets, its advantage over \quake grows as updates become more frequent, increasing from 65.5\% at U/Q=2:8 to 68.8\% at U/Q=5:5 and 83.8\% at U/Q=8:2. After the dataset size exceeds $10^6$, \ourLSM exhibits fluctuations due to its merge mechanism: throughput gradually declines as updates accumulate in buffers, but rises sharply after merges consolidate data. In particular, the major merge at $10^6$ vectors boosts throughput by 89\%/192\%/186\% on \sift/\deep/\gist, respectively.

In terms of recall, \ourLSM consistently maintains a high and stable recall (close to 0.997). In contrast, \quake exhibits a slight recall degradation on \deep (0.997 $\rightarrow$ 0.993), likely due to its cost-model-driven inverted-list splitting, which may occasionally result in insufficient scanning. Similarly, \spfresh shows noticeable recall fluctuations on \gist (0.999 $\rightarrow$ 0.991), potentially caused by IVF assignment drift under dynamic updates.

\vspace{-0.15in}
\begin{figure}[h]
    \centering
    \includegraphics[width=0.95\linewidth]{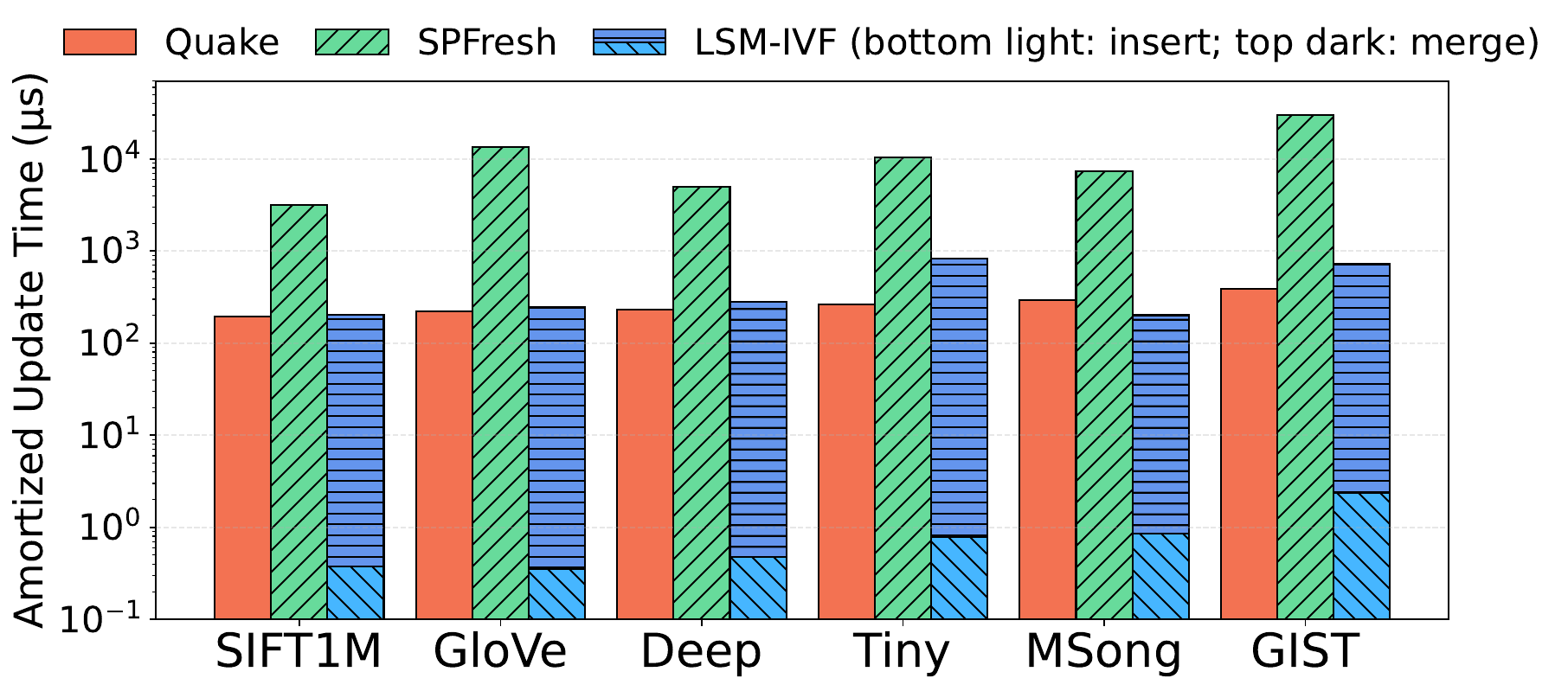}
    \vspace{-0.15in}
    \caption{Update efficiency.}
    \label{fig:update:performance}
\end{figure}
\noindent \textbf{Update Efficiency.}
Figure~\ref{fig:update:performance} compares the update time of \ourLSM with \quake and \spfresh.
\ourLSM decomposes updates into two components: (i) a lightweight online insertion into an in-memory buffer, and (ii) background merges triggered when thresholds are reached, whereas both baselines perform updates entirely online. The online insertion cost of \ourLSM is minimal due to the in-memory buffer design. 
When accounting for merge overhead, the amortized update cost of \ourLSM remains comparable to \quake, and significantly lower than \spfresh.
Note that during merging, queries need to access more levels before the merge completes, temporarily resulting in suboptimal query performance.
\spfresh incurs higher update latency as each insertion requires a graph traversal to locate neighbors and maintain connectivity. 


\subsection{In-Depth Analysis of \ourMethod}
\label{sec:exp-clip}

\noindent \textbf{Ablation Study.}
Figure~\ref{fig:Ablation} reports an ablation study of \ourIVF by selectively enabling pruning components. We consider two special settings: (i) inter-cluster pruning only, and (ii) intra-cluster pruning only (using $nprobe$-based cluster selection without inter-cluster pruning). 
Compared to the full \ourIVF, the inter-cluster-only variant exhibits an average 10.8\% QPS drop at 0.99 recall, while the intra-cluster-only variant leads to an average 20.4\% QPS drop, showing that removing either pruning strategy noticeably degrades the performance of \ourIVF. Compared to \ivfflat without any pruning strategies, the inter- and intra-cluster-only variants still achieve 10.7\%--29.0\% and 5.3\%--28.7\% higher QPS across datasets, respectively, showing effectiveness of both pruning strategies.

\begin{figure}[h]
    \centering
    \vspace{-0.15in}\includegraphics[width=\linewidth]{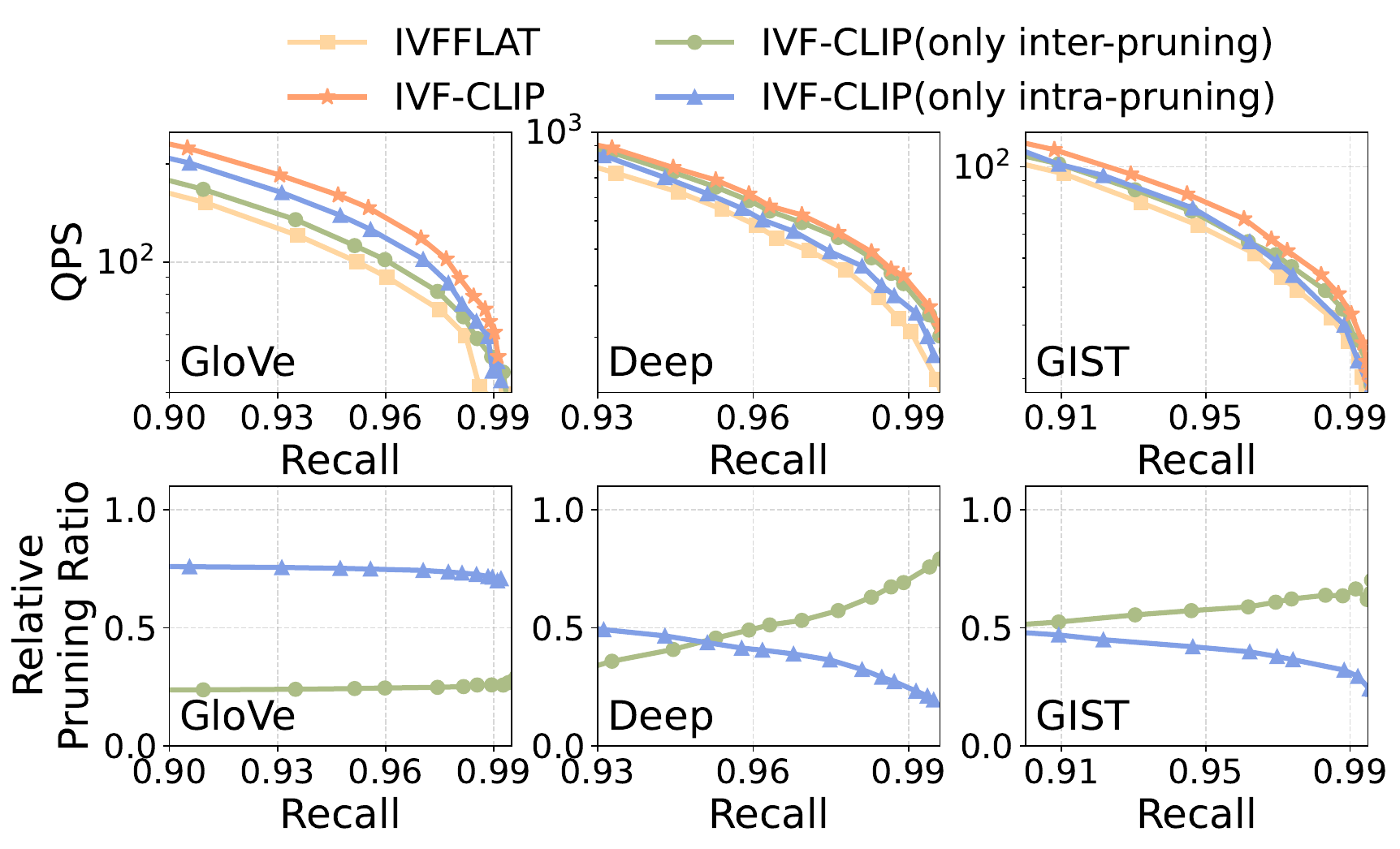}
    \vspace{-0.3in}
    \caption{Ablation study, where the relative pruning ratio is normalized by the pruning ratio of the full \ourIVF.}
    \label{fig:Ablation}
\end{figure}

The relative pruning ratio results show that the dominant pruning strategy varies across datasets. Inter-cluster pruning contributes an average of 65.6\% and 60.1\% of the full pruning effectiveness on \deep and \gist, respectively, making it the dominant factor on these datasets, while intra-cluster pruning contributes 74.2\% on \glove and plays a more important role there. These results demonstrate that both pruning strategies substantially contribute to the overall pruning effectiveness, and combining them consistently achieves the best performance. We omit the ablation of \ourHIVF due to its similar behavior to \ourIVF and space limitations.


\begin{figure}[h]
    \centering
    \vspace{-0.15in}\includegraphics[width=\linewidth]{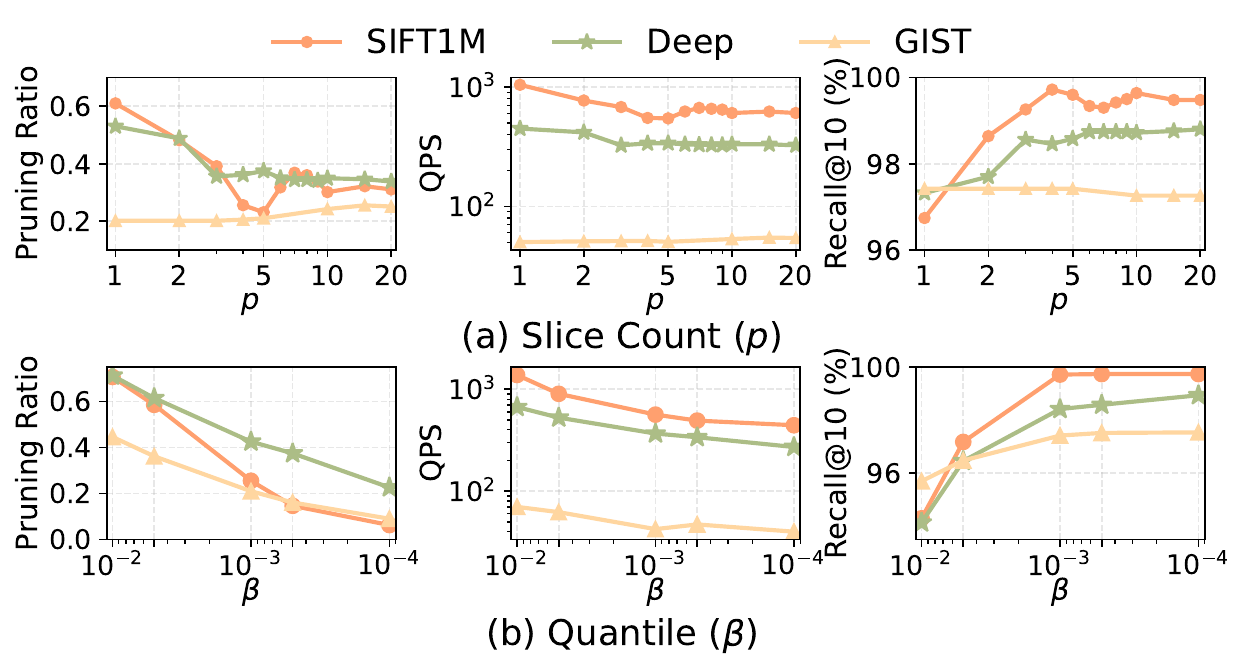}
    \vspace{-0.3in}
    \caption{Evaluating the effect of $p$ and $\beta$.}
    \label{fig:Parameter}
\end{figure}

\noindent \textbf{Parameter Sensitivity.} 
Figure~\ref{fig:Parameter}(a) examines the impact of $p$ (i.e., the number of slices) on pruning effectiveness and query quality under $nprobe=200$. When $p$ is small, the slicing is too coarse, leading to suboptimal bound quality and unstable pruning ratios and recall. As $p$ increases, the pruning behavior becomes progressively more stable and eventually plateaus, since finer slicing provides more accurate bounds for pruning. In practice, $p=20$ is sufficient to achieve near-stable pruning performance for most datasets.

Figure~\ref{fig:Parameter}(b) evaluates the effect of $\beta$ (i.e., the quantile used to set $\lambda$) on query performance under $nprobe=200$. Larger $\beta$ applies more aggressive pruning and improves efficiency but may hurt recall, whereas smaller $\beta$ is more conservative and preserves accuracy at the cost of reduced pruning. Overall, setting $\beta=0.001$ provides a good trade-off between efficiency and accuracy, achieving high recall while maintaining high pruning effectiveness.

\vspace{-0.1in}
\begin{figure}[h]
    \centering
    \includegraphics[width=\linewidth]{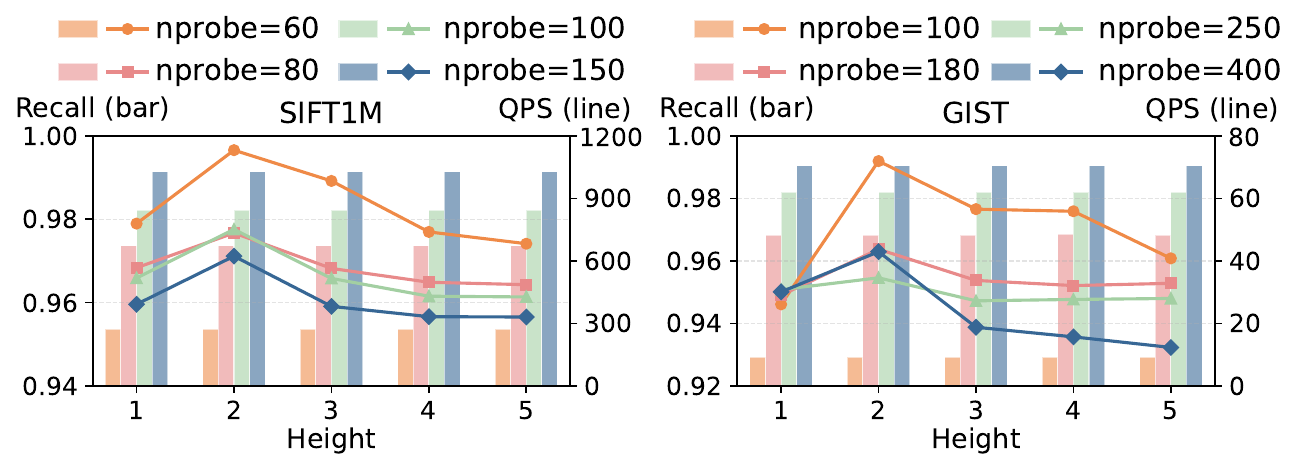}
    \vspace{-0.35in}
    \caption{Effect of hierarchy height in \ourHIVF.}
    \label{fig:height}
\end{figure}

Figure~\ref{fig:height} evaluates the effect of the hierarchy height in \ourHIVF. Overall, increasing the height has negligible impact on recall, while QPS exhibits a non-monotonic trend, i.e., first increasing and then decreasing, with the optimal height at 2. This behavior reflects a trade-off: additional layers reduce the cost of identifying the top-$nprobe$ clusters, but also incur extra traversal and routing overhead, which eventually outweighs the benefits and degrades performance. These results justify our default choice of using two layers, which achieves the best trade-off between efficiency and overhead.

\vspace{-0.1in}
\begin{figure}[h]
    \centering
    \includegraphics[width=\linewidth]{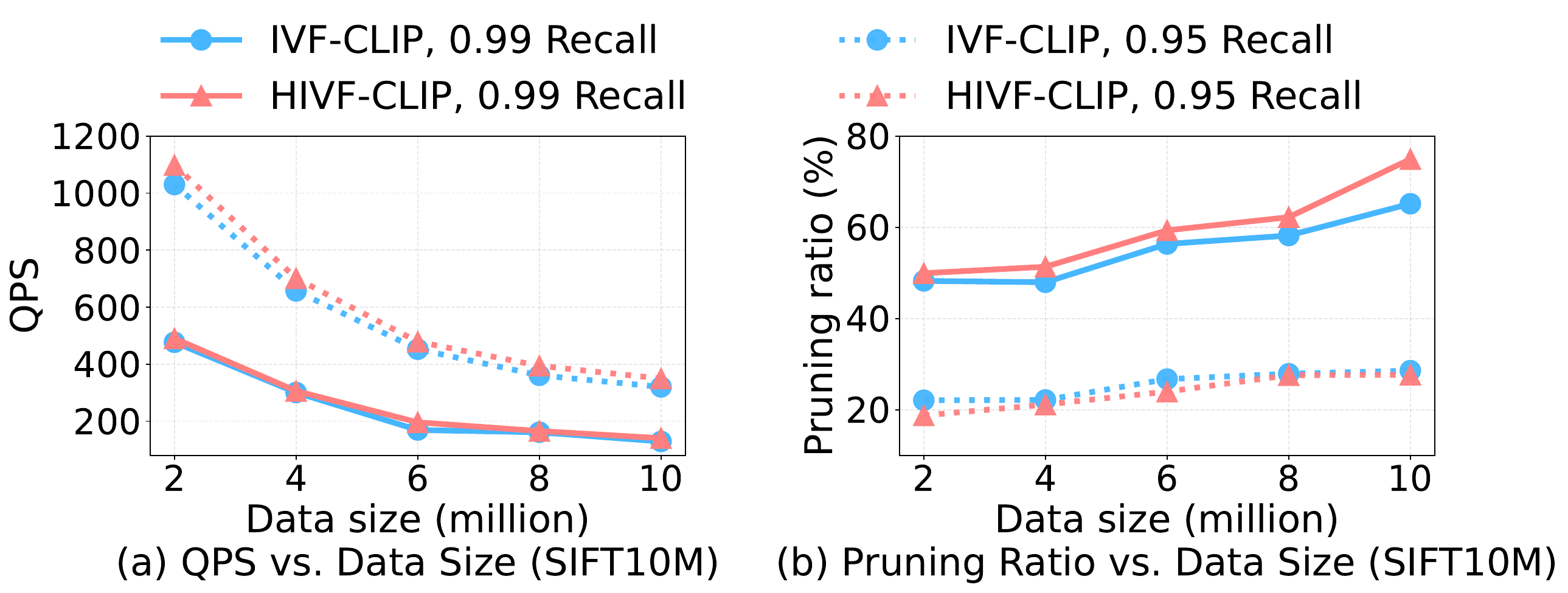}
    \vspace{-0.3in}
    \caption{Evaluating the scalability.}
    \label{fig:scalability}
\end{figure}

\noindent \textbf{Scalability Analysis.}
Figure~\ref{fig:scalability} evaluates scalability under two recall targets (0.95 and 0.99). As shown in Figure~\ref{fig:scalability}(a), when the dataset size increases from 2M to 10M vectors ($5\times$), QPS decreases sub-linearly at both recall levels. For example, at 0.99 recall, \ourIVF drops from about 480 to 130 QPS ($\sim$3.7$\times$), while \ourHIVF decreases from about 500 to 150 QPS ($\sim$3.3$\times$). A similar trend is observed at 0.95 recall.
Figure~\ref{fig:scalability}(b) shows that the pruning ratio also increases steadily with data size under both recall settings. 

\subsection{Evaluation in Disk-Resident Settings}
\label{sec:exp-extended}
This section evaluates \ourIVF in a disk-resident IVF setting. We extend \ourIVF to this setting by storing cluster centroids, inverted lists (vector identifiers only), vector--centroid distances, and $\lambda$ values in memory, while keeping the raw vectors on SSD. During query processing, the in-memory metadata supports both inter- and intra-cluster pruning, reducing SSD accesses and disk I/O.

Figure~\ref{fig:diskexp} compares \ourIVF with disk-based baselines (i.e., \spann and \ivfflat) on the 1-billion-scale \sift dataset. \ourIVF consistently outperforms both baselines across all recall levels. At 0.99 recall, it achieves 34.2\% higher QPS and 22.5\% lower I/O than \ivfflat. At 0.95 recall, the gains change to 50.2\% higher QPS and 18\% lower disk I/O. In contrast, \spann incurs 5--8$\times$ higher I/O due to frequent random accesses during graph traversal. The pruning ratio of \ourMethod exceeds 65\% at 0.99 recall, demonstrating that \ourIVF remains effective at billion scale under disk-resident settings.

\begin{figure}
    \centering
    \includegraphics[width=\linewidth]{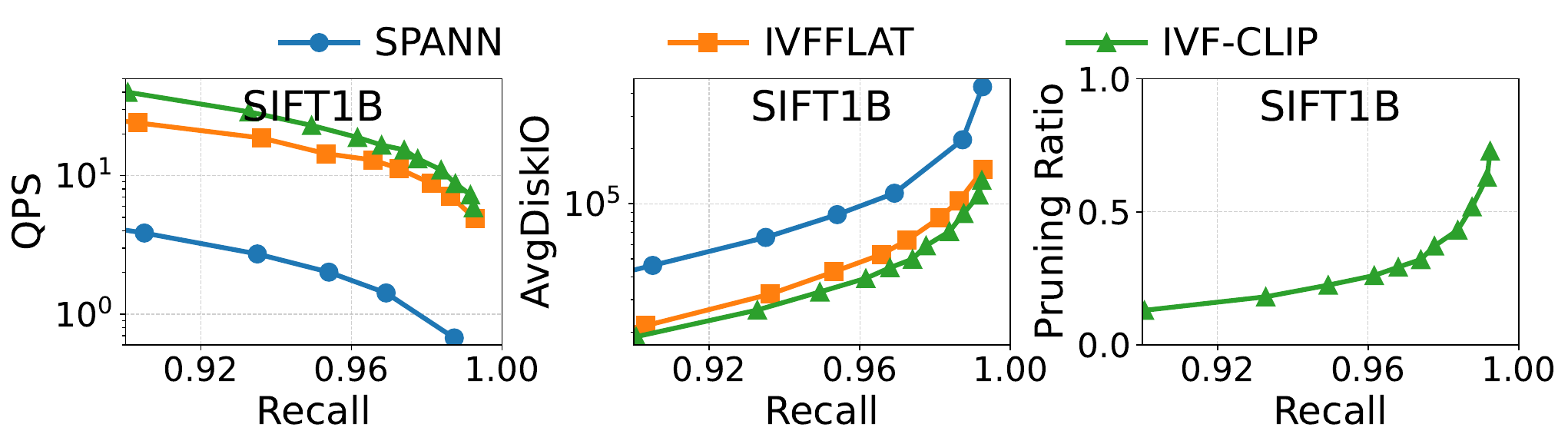}
    \vspace{-0.35in}
    \caption{Evaluation in the disk-resident setting.}
    \label{fig:diskexp}
\end{figure}



\section{Related Work}
\textbf{ANN Algorithms.}
ANN search methods are commonly categorized into four families: (1) partition-based~\cite{KDtree, coverTree, Mtree, RandomProj, ScalableNN, meta-faiss, spann, Quake, MicroNN}, (2) hash-based~\cite{QALSH, srs, idec, TODS, learnToHash}, (3) graph-based~\cite{HNSW, NSG, NSSG, tau-MNG, HVS, LSH-APG}, and (4) quantization-based approaches~\cite{pq, OPQ, RaBitQ, CacheLocality, vectorCompression, VQ, IMI}. This work focuses on IVF-based methods within the partition-based family. To reduce distance computation costs, numerous methods have been proposed~\cite{DCO, DCO2, LSH-APG, pq, RaBitQ, finger, trim, Tribase}. However, most of them either reduce the dimensionality of distance computations or rely on auxiliary structures (e.g., hash tables or quantization codes) for pruning, making them unsuitable for reducing cluster and vector accesses in pure IVF indexes. For example, TRIM~\cite{trim} exploits PQ-derived bounds for pruning and is therefore applicable only to IVFPQ rather than pure IVF indexes. Tribase~\cite{Tribase} is a notable exception that derives triangle-inequality-based bounds for IVF pruning. However, its pruning conditions are often overly conservative and incur considerable computational overhead. 

\noindent
\textbf{Dynamic IVF Index Maintenance.}
Under dynamic workloads, insertions and deletions may cause centroid drift and inverted-list imbalance, degrading IVF query performance. Existing methods maintain IVF indexes through online reorganization, such as repartitioning, split/merge, and cost-aware rebalancing~\cite{ada-ivf, Quake, spfresh}. Although effective, these maintenance operations compete with query processing for system resources, often reducing overall throughput. In contrast, LSM-based designs~\cite{lsm, bigtable, blsm} buffer updates and perform merges asynchronously, thereby improving throughput while maintaining query availability. Such designs have recently been extended to graph-based ANN indexes~\cite{LSM-VEC}. However, applying LSM-style designs to IVF indexes remains largely unexplored.

\section{Conclusion}
\label{sec:conclusion}
This paper presents \ourMethod, a lightweight cosine-law-based pruning technique for IVF-based vector search. We theoretically establish its correctness and pruning effectiveness, and develop two IVF variants, \ourIVF and \ourHIVF, to support both flat and hierarchical IVF structures. We further propose \ourLSM, which integrates \ourMethod with an LSM-tree to enable efficient updates while maintaining high query performance. Extensive experiments on static and dynamic workloads demonstrate the effectiveness of all proposed methods. 
While this work focuses on pure IVF indexes, extending \ourMethod to compressed IVF structures (e.g., IVFPQ) and combining it with other pruning techniques are promising directions for future research.



\clearpage

\balance
\bibliographystyle{ACM-Reference-Format}
\bibliography{reference}
\balance


\end{CJK*}
\end{document}